\documentclass[acmsmall,manuscript,screen]{acmart}

\AtBeginDocument{%
  }

\setcopyright{acmcopyright}
\copyrightyear{2018}
\acmYear{2018}
\acmDOI{XXXXXXX.XXXXXXX}

\usepackage{amsmath}
\usepackage{amsfonts}
\usepackage{amsthm}
\usepackage{url}
\usepackage{xspace}
\usepackage{soul}
\usepackage{xcolor}
\usepackage{algorithm}
\usepackage{algpseudocode}
\usepackage[inline]{enumitem}
\let\origtheorem\theorem
\let\origproof\proof
\usepackage{z-eves}
\let\theorem\origtheorem
\let\proof\origproof
\usepackage{tikz}
\usetikzlibrary{positioning,fit,calc,babel,
                shapes,arrows,automata,decorations.pathmorphing}

\newcommand{\true}{true}
\newcommand{\false}{false}
\newcommand{\lfun}{\longrightarrow}

\newcommand{\card}[1]{\lvert #1 \rvert}
\renewcommand{\plus}{\mathbin{\scriptstyle\sqcup}}
\renewcommand{\Cup}{un}
\newcommand{\disj}{\parallel}
\newcommand{\Size}{size}
\newcommand{\Nsize}{nsize}
\renewcommand{\Cap}{inters}
\newcommand{\Smin}{smin}
\newcommand{\Smax}{smax}
\newcommand{\e}{\emptyset}
\newcommand{\w}{\{\cdot \plus \cdot\}}
\renewcommand{\i}{[\cdot,\cdot]}
\newcommand{\Ncup}{nun}
\newcommand{\Ndisj}{\not\disj}

\newcommand{\CARD}{\card{\cdot}}
\newcommand{\LCARD}{\mathcal{L}_{\CARD}}
\newcommand{\STEPCARD}{\mathsf{solve\_size}}
\newcommand{\INT}{[\,]}
\newcommand{\LINT}{\mathcal{L}_{\INT}}
\newcommand{\setlog}{$\{log\}$\xspace}
\newcommand{\SATCARD}{\mathcal{SAT}_{\CARD}}
\newcommand{\SATINT}{\mathcal{SAT}_{\INT}}
\newcommand{\CLPSET}{CLP($\mathit{SET}$)\xspace}

\newcommand{\sSet}{\mathsf{Set}}
\newcommand{\sInt}{\mathsf{Int}}
\newcommand{\sUr}{\mathsf{Ur}}
\newcommand{\Var}{\mathcal{V}}
\newcommand{\Set}{\mathsf{S}}
\newcommand{\FSet}{\mathcal{F}_\Set}
\newcommand{\Int}{\mathsf{Z}}
\newcommand{\FInt}{\mathcal{F}_\Int}
\newcommand{\Ur}{\mathsf{U}}
\newcommand{\FUr}{\mathcal{F}_\Ur}
\newcommand{\PiSet}{\Pi_\mathsf{S}}
\newcommand{\PiInt}{\Pi_\Int}
\newcommand{\TINT}{\mathcal{T}_{\INT}}
\newcommand{\TInt}{\mathcal{T}_\Int}
\newcommand{\TUr}{\mathcal{T}_\Ur}
\newcommand{\CINT}{\mathcal{C}_{\INT}}
\newcommand{\FINT}{\Phi_{\INT}}

\newcommand{\iS}{\mathcal{R}}
\newcommand{\iF}[1]{(#1)^\iS}

\newcommand{\q}{\text{\normalfont\'{}}}
\newcommand{\qr}{\text{\normalfont\'{}}\hspace{5pt}}
\newcommand{\ql}{\hspace{5pt}\text{\normalfont\'{}}}

\newtheorem{example}{Example}
\newtheorem{definition}{Definition}
\newtheorem{remark}{Remark}
\newtheorem{theorem}{Theorem}
\newtheorem{lemma}{Lemma}

\newif\ifcomments
\commentstrue 

\begin{document}

\title{A Decision Procedure for a Theory of Finite Sets with Finite Integer Intervals}

\author{Maximiliano Cristi\'a}
\email{cristia@cifasis-conicet.gov.ar}
\affiliation{%
  \institution{Universidad Nacional de Rosario and CIFASIS}
  \city{Rosario}
  \country{Argentina}
}

\author{Gianfranco Rossi}
\email{gianfranco.rossi@unipr.it}
\affiliation{%
  \institution{Universit\`a di Parma}
  \city{Parma}
  \country{Italy}}

\begin{abstract}
In this paper we extend a decision procedure for the Boolean algebra of finite
sets with cardinality constraints ($\LCARD$) to a decision procedure for
$\LCARD$ extended with set terms denoting finite integer intervals ($\LINT$).
In $\LINT$ interval limits can be integer linear terms including
\emph{unbounded variables}. These intervals are a useful extension because they
allow to express non-trivial set operators such as the minimum and maximum of a
set, still in a quantifier-free logic. Hence, by providing a decision procedure
for $\LINT$ it is possible to automatically reason about a new class of
quantifier-free formulas. The decision procedure is implemented as part of the
\setlog (`setlog') tool. The paper includes a case study based on the elevator algorithm
showing that \setlog can automatically discharge all its invariance lemmas some
of which involve intervals.
\end{abstract}

\ccsdesc[500]{Theory of computation~Automated reasoning}
\ccsdesc[500]{Theory of computation~Constraint and logic programming}
\ccsdesc[500]{Theory of computation~Logic and verification}

\keywords{\setlog, set theory, integer intervals, decision procedure, constraint logic programming}

\received{20 February 2007}
\received[revised]{12 March 2009}
\received[accepted]{5 June 2009}

\maketitle

\section{Introduction}
In the context of formal verification and analysis, it is necessary to
discharge a number of verification conditions or proof obligations. Tools
capable of automating such proofs are essential to render the development
process cost-effective. At the base of proof automation is the concept of
decision procedure. Methods and tools based on set theory, such as B
\cite{Abrial00}, ProB \cite{Leuschel00} and Atelier-B \cite{atelierb}, provide
a proven vehicle for formal modeling, specification, analysis and verification
of software systems. These methods would benefit from new decision procedures
for fragments of set theory.

Departing from this observation, in this paper we provide a decision procedure for the Boolean algebra of finite sets extended with cardinality constraints and finite integer intervals or ranges\footnote{From now on, we will say integer intervals or just intervals meaning finite integer intervals.}.
To the best of our knowledge this is the first time this fragment of set theory is proven to be decidable.
The Boolean algebra of finite sets extended with cardinality constraints, denoted $\LCARD$ (read `l-card'), is known to be decidable since quite some time and solvers supporting it already exist \cite{DBLP:conf/frocos/Zarba02,Kuncak2006,Bansal2018,DBLP:journals/tplp/CristiaR23}.
The addition of integer intervals (i.e., $[k,m] = \{p \in \num | k \leq p \leq m\}$) allows to reason about operations such as the minimum of a set, the $i$-th smallest element of a set, partitioning of a set in the elements below and above a given number, etc.
Besides, integer intervals are important in the verification of programs with arrays \cite{DBLP:conf/vmcai/BradleyMS06} and static analysis \cite{DBLP:journals/tcs/SuW05}.
Although the other approaches just mentioned above could potentially reason about these set operations and solve problems in this field, the most challenging and original aspect of our approach is to obtain these results by successive extensions of a base theory of sets.

Indeed, the base theory of sets is the one supported by a Constraint Logic Programming (CLP) language known as \CLPSET \cite{Dovier00} providing a decision procedure for the Boolean algebra of \emph{hereditarily finite sets}, i.e., finitely nested sets that are finite at each level of nesting.
\CLPSET has been extended in several ways \cite{DBLP:journals/jar/CristiaR20,DBLP:journals/jar/CristiaR21a}.
In particular, $\LCARD$ extends \CLPSET with cardinality and linear integer constraints \cite{DBLP:journals/tplp/CristiaR23}.
In this paper we show how a further extension, called $\LINT$ (read `l-int'), can deal with integer intervals to provide a decision procedure.

\CLPSET and all of its extensions have been implemented in a tool called \setlog (read `setlog'), itself implemented in SWI-Prolog \cite{DBLP:journals/jlp/DovierOPR96,setlog}.
\setlog can be used as a CLP language and as a satisfiability solver.
This duality is reflected in the fact that \setlog code behaves as both a formula and a program.
Then, \setlog users can write programs and prove their properties using the same and only code and tool, all within set theory.
As a satisfiability solver, \setlog has proved to solve non-trivial problems \cite{DBLP:journals/jar/CristiaR21,DBLP:journals/jar/CristiaR21b}.
After presenting the implementation of $\LINT$ in \setlog we provide empirical evidence that it is useful in practice through a preliminary empirical evaluation and a case study based on the elevator algorithm.
The case study shows the duality formula-program that can be exploited in \setlog.

\paragraph{Structure of the paper}
The paper is structured as follows.
In Section \ref{motivation} we show an example where intervals combined with sets are needed to motivate this work.
Section \ref{language} presents a detailed account of the syntax and semantics of $\LINT$.
The algorithm to decide the satisfiability of $\LINT$ formulas, called $\SATINT$ (read `sat-int'), is described in Section \ref{satcard}.
$\SATINT$ is proved to be a decision procedure for $\LINT$ in Section \ref{decproc}.
In Section \ref{expressiveness} we provide several examples of non-trivial operations definable as $\LINT$ formulas as well as the kind of automated reasoning $\SATINT$ is capable of.
The implementation of $\SATINT$ as part of the \setlog tool is discussed in Section \ref{implementation}, including an initial empirical evaluation in Section \ref{empirical}.
In order to provide further empirical evidence that \setlog can reason about intervals combined with sets we present a case study based on the elevator problem in Section \ref{casestudy}.
In Section \ref{relwork} we put our work in the context of some results about interval reasoning.
Section \ref{conclusions} presents our conclusions.
Three appendices include technical details which are duly referenced throughout the paper.

\section{\label{motivation}Motivation}

Informally, $\LCARD$ is a language including the classic operators of set theory plus the cardinality operator.
Set operators are provided as constraints.
For example, $\Cup(A,B,C)$ is interpreted as $C = A \cup B$; $A \disj B$ is interpreted as $A \cap B = \emptyset$; and $\Size(A,j)$ as $\card{A} = j$.
$\LCARD$ includes so-called \emph{negative constraints} implementing the negation of constraints.
For example, $x \notin A$ corresponds to $\lnot x \in A$, and $\Ncup(A,B,C)$ corresponds to $\lnot C = A \cup B$.
In $\LCARD$ sets are finite and are built from the empty set ($\emptyset$) and a  set constructor, called \emph{extensional set}, of the form $\{x \plus A\}$ whose interpretation is $\{x\} \cup A$.
In $\{x \plus A\}$, both $x$ and $A$ can be variables.
Set elements can be integer numbers, some other ur-elements or sets.\footnote{Ur-elements (also known as atoms or individuals) are objects which contain no elements but are distinct from the empty set.}
Besides, $\LCARD$ includes linear integer terms and constraints.
Variables can denote integer numbers, ur-elements or sets.
Formulas are conjunctions and disjunctions of constraints.
$\LCARD$ is a decidable language whose solver is called $\SATCARD$ which has been implemented as part of the \setlog tool \cite{DBLP:journals/tplp/CristiaR23}.

$\LCARD$ can express, for example, that a set is the disjoint union of two sets of equal cardinality:
\begin{equation}
\Cup(A,B,C) \land A \disj B \land \Size(A,j) \land \Size(B,j)
\end{equation}
However it is unclear how $\LCARD$ can express that $C$ is an integer interval.
This is a limitation if $C$ is the collection of items numbered from $k$ to $m$ that are processed by agents $A$ and $B$ who should process an equal amount of them.
Clearly, the post-condition of such a system is:
\begin{equation}\label{eq:mot}
\Cup(A,B,[k,m]) \land A \disj B \land \Size(A,j) \land \Size(B,j)
\end{equation}
where $[k,m]$ denotes the set $\{p \in \num | k \leq p \leq m\}$, with $k$ and $m$ variables.
Hence, $\LCARD$ and $\SATCARD$ are not enough to describe such a system nor to automatically prove properties of it.

In this paper we show how $\LCARD$ and $\SATCARD$ can be extended with set terms of the form $[k,m]$ where either of the limits can be integer linear terms including unbounded variables---making $[k,m]$ a finite set.
The extensions are noted $\LINT$ and $\SATINT$, and $\SATINT$ is proved to be a decision procedure for $\LINT$ formulas.
One of the two key ideas behind this extension is the application of the following identity true of any set $A$ and any non-empty integer interval\footnote{Note that $[k,m]$ is not empty only when $k \leq m$.} $[k,m]$ (see the proof in Appendix \ref{app:proofs}):
\begin{equation}\label{eq:id}
k \leq m \implies (A = [k,m] \iff A \subseteq [k,m] \land \card{A} = m - k + 1)
\end{equation}
By means of this identity, \eqref{eq:mot} can be rewritten as follows\footnote{The case $m < k$ is not considered as it adds nothing to the understanding of the problem.}:
\begin{equation}\label{eq:mot-int}
\begin{split}
& \Cup(A,B,N) \land A \disj B \land \Size(A,j) \land \Size(B,j)
 \land N \subseteq [k,m] \land \Size(N,m - k + 1)
\end{split}
\end{equation}
where $N$ is a fresh variable. In this way \eqref{eq:mot-int} can be divided into a $\LCARD$ formula
\begin{equation}
\Cup(A,B,N) \land A \disj B \land \Size(A,j) \land \Size(B,j) \land \Size(N,m - k + 1)
\end{equation}
plus the constraint $N \subseteq [k,m]$.

Observe that in $\Cup(A,B,[k,m])$ in \eqref{eq:mot}, $A$ and $B$ can be variables and extensional sets built with the set constructor $\w$.
In general, $\LINT$ allows to express and reason about formulas where set variables, extensional sets and integer intervals can be freely combined. 
Integer intervals are a particular kind of set.
Even operations such as $[k,m] \setminus [i,j]$, whose result is not necessarily an interval, are dealt with correctly.
Furthermore, $\LINT$ can also deal with formulas where intervals are set elements, e.g., $[k,m] \in \{\{1,x,y,3\},\{-3,2,z\}\}$,
thanks to set unification \cite{Dovier2006}.
One of the keys for this result is to encode integer intervals in terms of the cardinality and integer constraints already provided by $\LCARD$.
In this sense, $\LINT$ takes a different direction than logics dealing only with intervals, e.g., \cite{DBLP:conf/padl/ErikssonP20}.

In general, any $\LINT$-formula $\Phi$ can be rewritten by means of \eqref{eq:id} into a conjunction of the form $\Phi_{\CARD} \land \Phi_{\subseteq[\,]}$ where $\Phi_{\CARD}$ is a $\LCARD$ formula and $\Phi_{\subseteq[\,]}$ is a conjunction of constraints of the form $X \subseteq [p,q]$ where $X$ is a variable and either $p$ or $q$ are variables.
We will refer to such intervals as \emph{variable-intervals}.
In this way, $\SATINT$ relies on $\SATCARD$ as follows.
If $\SATCARD$ finds $\Phi_{\CARD}$ unsatisfiable then $\Phi$ is unsatisfiable. 
However, if $\SATCARD$ finds $\Phi_{\CARD}$ satisfiable we still need to check if $\Phi_{\subseteq[\,]}$ does not compromise the satisfiability of $\Phi_{\CARD}$.
At this point the second key idea of our method comes into play.
First, $\SATINT$ asks $\SATCARD$ to compute a \emph{minimum solution} of $\Phi_{\CARD}$---roughly, a solution of $\Phi_{\CARD}$ where sets have the minimum number of elements \cite{DBLP:journals/tplp/CristiaR23}.
Second, if the computed minimum solution of $\Phi_{\CARD}$ is a solution of $\Phi$, then $\Phi$ is clearly satisfiable.
Third, we have proved that if any minimum solution of $\Phi_{\CARD}$ is not a solution of $\Phi$, then $\Phi$ is unsatisfiable.
That is, if any minimum solution of $\Phi_{\CARD}$ is not a solution of $\Phi$, then any larger solution\footnote{In this context, a larger solution is a solution where at least one of the sets involved in a minimum solution has at least one more element w.r.t. the cardinality of the minimum solution.} will not be a solution of $\Phi$.

In Section \ref{expressiveness}, we further study the power of $\LINT$ and $\SATINT$ by showing several problems they can describe and solve and in Section \ref{casestudy} we present a case study based on the elevator algorithm where we show how \setlog can automatically discharge all the invariance lemmas.

\section{\label{language}$\LINT$: a language for sets and integer intervals}


In this section we describe the syntax and semantics of our set-based language
$\LINT$. $\LINT$ is an extension of $\LCARD$. Although $\LCARD$ has been
thoroughly presented elsewhere \cite{DBLP:journals/tplp/CristiaR23} here we
reproduce that presentation with the extensions to integer intervals. Hence,
$\LINT$ is a multi-sorted first-order predicate language with three distinct
sorts: the sort $\sSet$ of all the terms which denote sets, the sort $\sInt$ of
terms denoting integer numbers, and the sort $\sUr$ of all the other terms.
Terms of these sorts are allowed to enter in the formation of set terms (in
this sense, the designated sets are hybrid), no nesting restrictions being
enforced (in particular, membership chains of any finite length can be
modeled). A handful of reserved predicate symbols endowed with a pre-designated
set-theoretic meaning is available. The usual linear integer arithmetic
operators are available as well. Formulas are built in the usual way by using
conjunction, disjunction and negation of atomic predicates. A few more complex
operators (in the form of predicates) are defined as $\LINT$ formulas, thus
making it simpler for the user to write complex formulas.

\subsection{Syntax}\label{syntax}

The syntax of the language is defined primarily by giving the signature upon
which terms and formulas are built.

\begin{definition}[Signature]\label{signature}
The signature $\Sigma_{\INT}$ of $\LINT$ is a triple $\langle
\mathcal{F},\Pi,\Var\rangle$ where:
\begin{itemize}
\item $\mathcal{F}$ is the set of function symbols along with their sorts, partitioned as
      $\mathcal{F} \defs \FSet \cup \FInt \cup \FUr$, where
      $\FSet \defs \{\e,\w,\i\}$, $\FInt = \{0,-1,1,-2,2,\dots\} \cup \{+,-,*\}$
      and $\FUr$ is a set of uninterpreted constant and function symbols.
\item $\Pi$ is the set of predicate symbols along with their sorts, partitioned as
$\Pi \defs \Pi_{=} \cup  \PiSet
\cup \Pi_{\Size} \cup \PiInt$, where $\Pi_{=} \defs  \{=,\neq\}$, $\PiSet \defs
\{\in,\notin,\Cup,\disj\}$, $\Pi_{\Size} \defs \{\Size\}$, and  $\PiInt \defs
\{\leq\}$.
%
\item $\Var$ is a denumerable set of variables partitioned as
$\Var \defs \Var_\Set \cup \Var_\Int \cup \Var_\Ur$.
\qed
\end{itemize}
\end{definition}

Intuitively, $\e$ represents the empty set; $\{x \plus A\}$ represents the set
$\{x\} \cup A$; $[m,n]$ represents the set $\{p \in \num | m \leq p \leq n\}$;
and $\Var_\Set$, $\Var_\Int$ and $\Var_\Ur$ represent sets of variables ranging
over sets, integers and ur-elements, respectively.

Sorts of function and predicate symbols are specified as follows: if $f$
(resp., $\pi$) is a function (resp., a predicate) symbol of arity $n$,
then its sort is an $n+1$-tuple $\langle s_1, \ldots ,s_{n+1} \rangle$ (resp.,
an $n$-tuple $\langle s_1, \ldots ,s_n \rangle$) of non-empty subsets of the
set of sorts $\{ \sSet , \sInt, \sUr \}$. This notion is denoted by $f:\langle
s_1, \ldots ,s_{n+1}\rangle$ (resp., by $\pi:\langle s_1, \ldots ,s_n\rangle
$). Specifically, the sorts of the elements of $\mathcal{F}$ and $\Var$ are the
following.

\begin{definition}[Sorts of function symbols and variables]\label{d:sorts}
The sorts of the symbols in $\mathcal{F}$ are as follows:
\begin{flalign*}
  \quad\quad & \e: \langle \{\sSet \} \rangle & \\
 & \mathsf{\w: \langle \{\sSet , \sInt, \sUr\}, \{ \sSet \} , \{ \sSet\}\rangle } & \\
 & \mathsf{\i: \langle \{\sInt\}, \{\sInt\}, \{ \sSet\}\rangle } & \\
 & c: \langle \{\sInt \} \rangle \text{, for any $c \in \{0,-1,1,-2,2,\dots\}$} & \\
 & \cdot + \cdot, \cdot - \cdot, \cdot * \cdot:
  \langle \{\sInt\}, \{\sInt\} , \{\sInt\}\rangle & \\
 & f: \langle \{\sUr\},\dots ,\{\sUr\} \rangle \in (\{\sUr\})^{n+1}\text{, if $f
\in \FUr$ is of arity $n \ge 0$}. &
 \end{flalign*}
The sorts of variables are as follows:
\begin{flalign*}
  \quad\quad & v: \langle \{\sSet \} \rangle \text{, if $v \in \Var_\Set$} & \\
 & v: \langle \{\sInt \} \rangle \text{, if $v \in \Var_\Int$} & \\
 & v: \langle \{\sUr \} \rangle \text{, if $v \in \Var_\Ur$} & \tag*{\qed}
 \end{flalign*}
\end{definition}

\begin{definition}[Sorts of predicate symbols]\label{d:sorts_pred}
The sorts of the predicate symbols in $\Pi$ are as follows
(symbols $\Cup$ and $\Size$ are prefix; all other symbols in $\Pi$ are infix):
\begin{flalign*}
 \quad\quad & =,\neq: \langle \{\sSet , \sInt, \sUr \}, \{ \sSet , \sInt, \sUr \} \rangle  & \\
 & \in,\notin: \langle \{\sSet, \sInt, \sUr \} , \{\sSet \} \rangle & \\
 & \Cup: \langle \{\sSet \} , \{\sSet \}, \{\sSet \} \rangle & \\
 & \disj: \langle \{\sSet \} , \{\sSet \} \rangle & \\
 & \Size: \langle \{\sSet \} , \{\sInt \} \rangle & \\
 & \leq: \langle \{\sInt \} , \{\sInt \} \rangle & \tag*{\qed}
 \end{flalign*}
\end{definition}

Note that arguments of $=$ and $\neq$ can be of any of the three considered
sorts. We do not have distinct symbols for different sorts, but the
interpretation of $=$ and $\neq$ (see Section \ref{semantics}) depends on the
sorts of their arguments.

The set of admissible (i.e., well-sorted) $\LINT$ terms is defined as follows.

\begin{definition}[$\INT$-terms]\label{LCARD-terms}
The set of \emph{$\INT$-terms}, denoted by $\TINT$, is the minimal subset of
the set of $\Sigma_{\INT}$-terms generated by the following grammar complying
with the sorts as given in Definition \ref{d:sorts}:
\begin{flalign*}
 \quad\quad C     &::=   0 \hspace{2pt}|\hspace{2pt}
          {-1} \hspace{2pt}|\hspace{2pt}
            1 \hspace{2pt}|\hspace{2pt}
          {-2} \hspace{2pt}|\hspace{2pt}
            2 \hspace{2pt}|\hspace{2pt}
           \dots & \\
\TInt &::= C \hspace{2pt}|\hspace{2pt}
          \Var_\Int \hspace{2pt}|\hspace{2pt}
          C * \Var_\Int \hspace{2pt}|\hspace{2pt}
          \Var_\Int * C \hspace{2pt}|\hspace{2pt}
          \TInt + \TInt \hspace{2pt}|\hspace{2pt}
          \TInt - \TInt & \\
\TINT & ::=
  \TInt \hspace{2pt}|\hspace{2pt}
  \TUr \hspace{2pt}|\hspace{2pt} \Var_\Ur \hspace{2pt}|\hspace{2pt}
  \mathit{Set} & \\
\mathit{Set} & ::=
   \q\e\qr
   \hspace{2pt}|\hspace{2pt}
   \Var_S
   \hspace{2pt}|\hspace{2pt}
      \q\{\qr \TINT
           \ql\hspace{-2pt}\plus\hspace{-2pt}\qr \mathit{Set} \ql\}\q
   \hspace{2pt}|\hspace{2pt}
      \q[\qr \TInt
           \ql\hspace{-2pt},\hspace{-2pt}\qr \TInt \ql]\q &
\end{flalign*}
where $\TInt$ (resp., $\TUr$) represents any non-variable $\FInt$-term
(resp., $\FUr$-term). \qed
\end{definition}

As can be seen, the grammar allows only integer linear terms.

If $t$ is a term $f(t_1,\dots,t_n)$, $f \in \mathcal{F}, n \ge 0$, and $\langle
s_1, \ldots ,s_{n+1} \rangle$ is the sort of $f$, then we say that $t$ is of
sort $\langle s_{n+1} \rangle$. The sort of any $\INT$-term $t$ is always
$\langle \{\sSet\}\rangle$ or $\langle \{\sInt\} \rangle$ or $\langle \{\sUr\}
\rangle$. For the sake of simplicity, we say that $t$ is of sort
$\sSet$ or $\sInt$ or $\sUr$, respectively. In particular, we say that a
$\INT$-term of sort $\sSet$ is a {\em set term}, that set terms of the form
$\w$ are \emph{extensional} set terms, and set terms of the form $\i$ are
integer intervals or just intervals. The first parameter of an extensional set
term is called \emph{element part} and the second is called \emph{set part}.
Observe that one can write terms representing sets which are nested at any
level. The parameters of intervals are called \emph{left} and \emph{right
limits}, respectively. It is important to remark that interval limits can be
integer linear terms including variables.

Hereafter, we will use the following notation for extensional set terms:
$\{t_1,t_2,\dots,t_n \plus t\}$ is a shorthand for $\{t_1 \plus \{t_2 \,\plus\,
\cdots \{ t_n \plus t\}\cdots\}\}$, while $\{t_1,t_2,\dots,t_n\}$ is a
shorthand for $\{t_1,t_2,\dots,t_n \plus \e\}$. Moreover, we will use the
following name conventions: $A, B, C, D$
for terms of sort $\sSet$ excluding integer intervals;  $i, j, k, m, p, q$
for terms of sort $\sInt$; $a, b, c, d$ for terms of sort $\sUr$; and $x, y, z$
for terms of any of the three sorts.

\begin{example}[Set terms]
The following $\Sigma_{\INT}$-terms are set terms:
\begin{flalign*}
  \quad\quad & \e &\\
 & \{x \plus A\} & \\
 & \{4+k,f(a,b)\}, \text{ i.e., } \{4+k \plus \{f(a,b) \plus \e \}\},
  \text{ where $f$ is a (uninterpreted) } & \\
 & \text{symbol in $\FUr$} & \\
 & [-3,2*m] &
\end{flalign*}
On the opposite, $\{x \plus 17\}$ is not a set term. \qed
\end{example}


The sets of well-sorted $\LINT$ constraints and formulas are defined as
follows.

\begin{definition}[$\INT$-constraints]\label{primitive-constraint}
If $\pi \in \Pi$ is a predicate symbol of sort $\langle s_1, \ldots , s_n
\rangle$, and for each $i=1,\ldots , n$, $t_i$ is a $\INT$-term of sort
$\langle s'_i \rangle$ with $s'_i \subseteq s_i$, then $\pi (t_1,\ldots ,t_n)$
is a \emph{$\INT$-constraint}. The set of $\INT$-constraints is denoted by
$\CINT$. \qed
\end{definition}

\begin{example}
If $k,m \in \Var_\Int$ and $A \in \Var_\Set$, then $[k,m] = \{3
\plus A\}$ is a $\INT$-constraint but $[k,\e] = \{3 \plus A\}$ is not, as
$[k,\e]$ is not a $\INT$-term because $\e$ does not belong to $\TInt$.
\qed
\end{example}

$\INT$-constraints whose arguments are of sort $\sSet$ (including
$\Size$-constraints) will be called \emph{set constraints}; $\INT$-constraints
whose arguments are of sort $\sInt$ will be called \emph{integer constraints}.

\begin{definition}[$\INT$-formulas]\label{formula}
The set of $\INT$-formulas, denoted by $\FINT$, is given
by the following grammar:
\begin{flalign*}
 \quad\quad & \FINT ::=
  \true \mid \false \mid \CINT \mid \FINT \land \FINT \mid \FINT \lor
  \FINT &
\end{flalign*}
where $\CINT$ represents any element belonging to the set of
$\INT$-constraints. \qed
\end{definition}

\begin{example}[$\INT$-formulas]\label{ex:formulas}
The following are $\INT$-formulas:
\begin{flalign*}
 \quad\quad & a \in [k,i] \land a \notin B \land \Cup([m,k+1],B,C) \land C = \{x \plus D\} & \\
 & \Cup(A,B,[k,j]) \land n+k > 5 \land \Size([k,j],n) \land B \neq \e &
\end{flalign*}
On the contrary, $\Cup(A,B,23)$ is not a $\INT$-formula because $\Cup(A,B,23)$
is not a $\INT$-constraint ($23$ is not of sort $\sSet$ as required by the
sort of $\Cup$). \qed
\end{example}

\begin{remark}
$\LCARD$ coincides with $\LINT$ without integer
interval terms. That is, if the function symbol $\i$ is removed from $\FSet$
(and consequently from the other definitions of the elements of the language),
we get $\LCARD$. \qed
\end{remark}

\subsection{\label{semantics}Semantics}

Sorts and symbols in $\Sigma_{\INT}$ are interpreted according to the
interpretation structure $\iS \defs \langle D,\iF{\cdot}\rangle$, where $D$ and
$\iF{\cdot}$ are defined as follows.

\begin{definition} [Interpretation domain] \label{def:int_dom}
The interpretation domain $D$ is partitioned as $D \defs D_\sSet \cup D_\sInt
\cup D_\sUr$ where:
\begin{itemize}
\item $D_\sSet$ is the set of all hereditarily finite hybrid
sets built from elements in $D$. Hereditarily finite sets are those sets that
admit (hereditarily finite) sets as their elements, that is sets of sets.
\item $D_\sInt$ is the set of integer numbers, $\mathbb{Z}$.
\item $D_\sUr$ is a collection of other objects. \qed
\end{itemize}
\end{definition}

\begin{definition} [Interpretation function] \label{app:def:int_funct}
The interpretation function $\iF{\cdot}$ is defined as follows:
\begin{itemize}
\item Each sort $\mathsf{X} \in \{\sSet,\sInt,\sUr\}$ is mapped to
      the domain $D_\mathsf{X}$.

\item The constant and function symbols in $\mathcal{F}_\Set$ are
interpreted as follows:
  \begin{itemize}
  \item $\e$ is interpreted as the empty set
  \item $\{ x \plus A \}$ is interpreted as the set $\{x^\iS\} \cup A^\iS$
  \item $[k,m]$ is interpreted as the set $\{p \in \num | k^\iS \leq p \leq m^\iS\}$
  \end{itemize}

\item The constant and function symbols in $\mathcal{F}_\Int$ are
interpreted as follows:
\begin{itemize}
\item Each element in \{0,-1,1,-2,2,\dots\} is interpreted as the
corresponding integer number
\item $i + j$ is interpreted as $i^\iS + j^\iS$
\item $i - j$ is interpreted as $i^\iS - j^\iS$
\item $i * j$ is interpreted as $i^\iS * j^\iS$
\end{itemize}

\item The predicate symbols in $\Pi$ are interpreted as follows:
  \begin{itemize}
   \item $x = y$, where $x$ and $y$ have the same sort $\mathsf{X}$,
   $\mathsf{X} \in \{\sSet,\sInt,\sUr\}$, is interpreted as the
   identity between $x^\iS$ and $y^\iS$ in $D_\mathsf{X}$;
   conversely, if $x$ and $y$ have different sorts, $x = y$ is interpreted in
   such a way as its truth value will be always $\false$
   \item $x \in A$ is interpreted as $x^\iS \in A^\iS$
   \item $\Cup(A,B,C)$ is interpreted as $C^\iS = A^\iS \cup B^\iS$
   \item $A \disj B$ is interpreted as $A^\iS \cap B^\iS = \emptyset$
   \item $\Size(A,k)$ is interpreted as $\card{A^\iS} = k^\iS$
   \item $i \leq j$ is interpreted as $i^\iS \leq j^\iS$
   \item $x \neq y$ and $x \notin A$ are
   interpreted as $\lnot x = y$ and $\lnot x \in A$, respectively.
\qed
\end{itemize}
\end{itemize}
\end{definition}

Note that integer intervals in $\LINT$ denote always finite sets given that
their limits can assume only integer values---in other words, integer limits
cannot be $\pm\infty$.

The interpretation structure $\iS$ is used to evaluate each $\INT$-formula
$\Phi$ into a truth value $\Phi^\iS = \{\true,\false\}$ in the following way:
set constraints (resp., integer constraints) are evaluated by $\iF{\cdot}$
according to the meaning of the corresponding predicates in set theory (resp.,
in number theory) as defined above; $\INT$-formulas are evaluated by
$\iF{\cdot}$ according to the rules of propositional logic. In particular,
observe that equality between two set terms is interpreted as the equality in
$D_\sSet$; that is, as set equality between hereditarily finite hybrid sets.
Such equality is regulated by the standard \emph{extensionality axiom}, which
has been proved to be equivalent, for hereditarily finite sets, to the
following equational axioms \cite{Dovier00}:
\begin{gather}
\{x, x \plus A\} = \{x \plus A\} \tag{$Ab$} \label{Ab} \\
\{x, y \plus A\} = \{y, x \plus A\} \tag{$C\ell$} \label{Cl}
\end{gather}
Axiom \eqref{Ab} states that duplicates in a set term do not matter
(\emph{Absorption property}). Axiom \eqref{Cl} states that the order of
elements in a set term is irrelevant (\emph{Commutativity on the left}). These
two properties capture the intuitive idea that, for instance, the set terms
$\{1,2\}$, $\{2,1\}$, and $\{1,2,1\}$ all denote the same set.

A \emph{valuation} $\sigma$ of a formula $\Phi$ is an assignment of
values from $\mathcal{D}$ to the free variables of $\Phi$ which respects the sorts
of the variables. $\sigma$ can be extended to terms in a straightforward
manner. In the case of formulas, we write $\Phi[\sigma]$ to denote the
application of a valuation to a formula $\Phi$. $\sigma$ is a \emph{successful
valuation} (or, simply, a \emph{solution}) if $\Phi[\sigma]$ is true in $\iS$.

\subsection{\label{derived}{Derived Constraints}}

$\LINT$ can be extended to support other set and integer operators definable by
means of suitable $\LINT$ formulas.

Dovier et al. \cite{Dovier00} proved that the collection of predicate symbols
in $\Pi_{=} \cup \PiSet$ is sufficient to define constraints implementing the
set operators $\cap$, $\subseteq$ and $\setminus$. For example, $A \subseteq B$
can be defined by the $\LINT$ formula $\Cup(A,B,B)$. In a similar fashion,
$\{=,\neq\} \cup \PiInt$ is sufficient to define $<$, $>$ and $\geq$. With a
little abuse of terminology, we say that the set  and integer predicates that
are given as $\INT$-formulas are \emph{derived constraints}. In Section
\ref{expressiveness}, we introduce more complex derived constraints that can be
written only when integer intervals are available.

Whenever a formula contains a derived constraint, the constraint is replaced by
its definition turning the given formula into a $\LINT$ formula. Precisely, if
formula $\Phi$ is the definition of constraint $c$, then $c$ is replaced by
$\Phi$ and the solver checks satisfiability of $\Phi$ to determine
satisfiability of $c$. Thus, we can completely ignore the presence of derived
constraints in the subsequent discussion about constraint solving and formal
properties of our solver.


The negated versions of set and integer operators can be introduced as derived
constraints, as well. The derived constraint for $\lnot\cup$ and
$\lnot\disj$ (called $\Ncup$ and $\Ndisj$, respectively) are shown in
\cite{Dovier00}.
For example, $\lnot(A \cup B = C)$ is introduced as:
\begin{equation}\label{e:nun}
\Ncup(A,B,C) \defs
    (n \in C \land n \notin A \land n \notin B)
     \lor (n \in A \land n \notin C)
     \lor (n \in B \land n \notin C)
\end{equation}
With a
little abuse of terminology, we will refer to these predicates as
\emph{negative constraints}.

Thanks to the availability of negative constraints, (general) logical negation
is not strictly necessary in $\LINT$.

Now that we have derived and negative constraints it is easy to see that
$\LINT$ expresses the Boolean algebra of sets with cardinality extended with
integer intervals. That is, one can write Boolean formulas where arguments are
extensional sets and integer intervals.

\begin{remark}
From now on, we will use $A \subseteq B$ as a synonym of the $\LINT$
constraint $\Cup(A,B,B)$. In particular we will write $X \subseteq [k,m]$ in
place of $\Cup(X,[k,m],[k,m])$. \qed
\end{remark}

\begin{remark}[A stack of constraint languages and solvers]
As we have said, $\LINT$ and $\SATINT$ are built on top of $\LCARD$ and
$\SATCARD$. In turn, $\LCARD$ and $\SATCARD$ are built on top of \CLPSET
\cite{Dovier00}. \CLPSET is based on a constraint language including
$\mathcal{F}_\Set$ and $\PiSet$; formulas in \CLPSET
are built as in $\LCARD$. One of the main concepts behind  \CLPSET is \emph{set
unification} \cite{Dovier2006}. $\LCARD$ effectively extends \CLPSET by
introducing $\Size$ and integer constraints; and $\LINT$ effectively extends
$\LCARD$ by admitting integer intervals. Set unification goes all the way up to
$\SATINT$; it is also pervasive in other CLP schemas developed by the authors
\cite{DBLP:journals/jar/CristiaR20,DBLP:journals/jar/CristiaR21a}. \setlog implements $\LCARD$
and $\SATCARD$, while the extension of $\setlog$ to implement $\LINT$ is
presented in this paper. \qed
\end{remark}

\section{\label{satcard}$\SATINT$: a constraint solving procedure for $\LINT$}


In this section, we show how $\SATCARD$ can be further extended to support set constraints  whose arguments can be integer intervals.
The resulting constraint solving procedure, i.e., $\SATINT$, is a decision procedure for $\LINT$ formulas.
Furthermore, it produces a finite representation of all the possible solutions of any satisfiable $\LINT$ formula (see Section \ref{decproc}).

\subsection{The solver}

The overall organization of $\SATINT$ is shown in Algorithm \ref{glob}.
Basically, $\SATINT$ uses four routines: \textsf{gen\_size\_leq},
$\mathsf{STEP_{S\INT}}$ (called from $\mathsf{step\_loop}$\footnote{As
$\mathsf{step\_loop}$ merely loops  calling $\mathsf{STEP_{S\INT}}$, we will
talk about the latter rather than the former. $\mathsf{STEP_{S\INT}}$ is the
key procedure in Algorithm \ref{glob}.}), \textsf{remove\_neq} and $\STEPCARD$.

\newcommand{\minsol}{\mathsf{check\_all\_minsol}}

\begin{algorithm}
\begin{algorithmic}[0]
 \State $\Phi \gets \textsf{gen\_size\_leq}(\Phi)$;
 \Repeat
   \State $\Phi' \gets \Phi$;
   \State $\Phi \gets \textsf{remove\_neq}(\mathsf{step\_loop}(\Phi))$
 \Until{$\Phi = \Phi'$; \hfill{\footnotesize[end of main loop]}}
 \State \textbf{let} $\Phi$ \textbf{be} $\Phi_{\CARD}
      \land \Phi_{\subseteq\INT}$;
 \State \textbf{let} $\Phi_{\CARD}$ \textbf{be} $\Phi_1 \land \Phi_2$;
 \If{$\Phi_{\subseteq\INT} \neq \true$}
   \State\Return{$\minsol(\Phi_1,\Phi_2,\Phi_{\subseteq\INT})$}
 \Else
 \State\Return{$\STEPCARD(\Phi_1) \land \Phi_2$}
 \EndIf
\vspace{2mm}
\Procedure{$\mathsf{step\_loop}$}{$\Phi$}
   \Repeat
     \State $\Phi' \gets \Phi$;
     \State $\Phi \gets \mathsf{STEP_{S\INT}}(\Phi)$
     \hfill{\footnotesize[$\mathsf{STEP_{S\INT}}$ is a key procedure]}
   \Until{$\Phi = \Phi'$}
 \State\Return{$\Phi$}
\EndProcedure
\vspace{2mm}
\Procedure{$\minsol$}{$\Phi_1,\Phi_2,\Phi_{\subseteq\INT}$}
 \If{$\STEPCARD(\Phi_1,Min)$}
 \For{$m_1 = c_1, \dots, m_k=c_k$ solution of $0 \leq m_1 \leq Min \land \dots \land 0 \leq m_k \leq Min \land Min = \sum_{i=1}^k m_i$}
    \If{$\mathsf{step\_loop}(\Phi_1 \land \Phi_2 \land \Phi_{\subseteq\INT} \land m_1 = c_1 \land \dots \land m_k=c_k) \neq \false$}
       \State\Return{$\Phi_1 \land \Phi_2 \land \Phi_{\subseteq\INT}$}
    \EndIf
    \EndFor
    \State\Return{$\false$}
 \Else
    \State\Return{$\false$}
 \EndIf
\EndProcedure
\end{algorithmic}
\caption{The solver $\SATINT$. $\Phi$ is the input formula.} \label{glob}
\end{algorithm}

\textsf{gen\_size\_leq} simply adds integer constraints to the input formula $\Phi$ to force the second argument of $\Size$-constraints in $\Phi$ to be non-negative integers.
$\mathsf{STEP_{S\INT}}$ includes the constraint solving procedures for the $\LCARD$ fragment as well as the new constraint solving procedures for set constraints whose arguments are intervals.
$\mathsf{STEP_{S\INT}}$ applies specialized rewriting procedures to the current formula $\Phi$ and returns either $\false$ or the modified formula. Each rewriting procedure applies a few non-deterministic rewrite rules which reduce the syntactic complexity of $\INT$-constraints of one kind.
\textsf{remove\_neq} deals with the elimination of $\neq$-constraints involving set variables.
Its motivation and definition will be made evident later in Section \ref{inequalities}.
$\STEPCARD$ is the adaptation of the decision procedure proposed by C. Zarba for cardinality constraints \cite{DBLP:conf/frocos/Zarba02} to our CLP framework.
Both $\mathsf{STEP_{S\INT}}$ and $\STEPCARD$ use the SWI-Prolog CLP(Q) library to solve linear integer arithmetic problems \cite{holzbaur1995ofai}.
These problems may be part of $\Phi$ or they are generated during set processing.
Besides, $\STEPCARD$ uses a SAT solver implemented in Prolog by Howe and King \cite{DBLP:journals/tcs/HoweK12} to help in the implementation of Zarba's algorithm.
$\minsol$ basically iterates over all the minimum solutions. The precise implementation of this procedure is discussed in Section \ref{after}.

The execution of $\mathsf{STEP_{S\INT}}$ and \textsf{remove\_neq} is iterated
until a fixpoint is reached, i.e., the formula is irreducible. These routines
return $\false$ whenever (at least) one of the procedures in it rewrites $\Phi$
to $\false$. In this case, a fixpoint is immediately detected.

As we will show in Section \ref{decproc}, when all the non-deterministic
computations of $\SATINT(\Phi)$ return $\false$, then we can conclude that
$\Phi$ is unsatisfiable; otherwise, when at least one of  them does not return
$\false$, then we can conclude that $\Phi$ is satisfiable and each solution of
the formulas returned by $\SATINT$ is a solution of $\Phi$, and vice versa.

Apart from the new rewrite rules, $\SATINT$ works exactly as $\SATCARD$ up until the end of the main loop.
After the main loop, $\SATINT$ differs from $\SATCARD$ in:
\begin{enumerate*}
\item dividing $\Phi$ into $\Phi_{\CARD}$ and $\Phi_{\subseteq[\,]}$; and
\item adding the \textbf{then} branch.
\end{enumerate*}
$\Phi_{\CARD}$ is a $\LCARD$ formula; and $\Phi_{\subseteq[\,]}$ is a conjunction of constraints of the form $X \subseteq [p,q]$ where $[p,q]$ is a variable-interval.
In turn, $\Phi_{\CARD}$ is divided into $\Phi_1$ and $\Phi_2$ as in $\SATCARD$: $\Phi_1$ contains all the integer constraints and all the $\Cup$, $\disj$ and $\Size$ constraints, and $\Phi_2$ is the rest of $\Phi_{\CARD}$ (i.e., $\notin$-constraints, and $=$ and $\neq$ constraints not involving integer terms).
If $\Phi_{\subseteq[\,]}$ is $\true$, $\SATINT$ executes the \textbf{else} branch which corresponds to the implementation of $\SATCARD$.
This means that when in $\Phi$ there are no integer intervals, $\SATINT$ reduces to $\SATCARD$.
The \textbf{then} branch is exclusive of $\SATINT$ and is entered only if at this point a constraint of the form $X \subseteq [p,q]$ is still present in $\Phi$.
However, $\Phi_{\subseteq[\,]}$ is not passed in to $\STEPCARD$, because Zarba's algorithm cannot deal with integer intervals, but instead $\minsol$ is called to make a final satisfiability judgment as is explained in Section \ref{after}.

\subsection{Rewrite rules of \CLPSET and $\SATCARD$}

The rewrite rules used by $\SATINT$ are defined as follows.

\begin{definition}[Rewrite rules]\label{d:rw_rules}
If $\pi$ is a symbol in $\Pi$ and $\phi$ is a $\INT$-constraint based on $\pi$,
then a \emph{rewrite rule for $\pi$-constraints} is a rule of the form $\phi
\lfun \Phi_1 \lor \dots \lor \Phi_n$, where $\Phi_i$, $1 \leq i \leq n$, are
$\INT$-formulas. Each $\Sigma_{\INT}$-predicate matching $\phi$ is
non-deterministically rewritten to one of the $\Phi_i$. Variables appearing in
the right-hand side but not in the left-hand side are assumed to be fresh
variables, implicitly existentially quantified over each $\Phi_i$. \qed
\end{definition}

A \emph{rewriting procedure} for $\pi$-constraints consists of the collection
of all the rewrite rules for $\pi$-constraints. The first
rule whose left-hand side matches the input $\pi$-constraint is used to rewrite
it.
Constraints that are rewritten by no rule are called \emph{irreducible}.
Irreducible constraints are part of the final answer of $\mathsf{STEP_{S\INT}}$
(see Definition \ref{def:solved}).

The following conventions are used throughout the rules. $\dot x$, for any name
$x$, is a shorthand for $x \in \Var$, i.e., $\dot x$ represents a variable. In
particular, variable names $\dot n$, $\dot n_i$, $\dot N$ and $\dot{N_i}$
denote fresh variables of sort $\sInt$ and $\sSet$, respectively.

Before introducing the new rewrite rules added to $\SATINT$ we show some of the
key rewrite rules inherited from \CLPSET and $\SATCARD$ (see Figure
\ref{f:clpset}). Without the new rewrite rules, $\SATINT$ can deal with $\LINT$
formulas as long as integer intervals are not present.

Rule \eqref{eq:ext} is
the main rule of set unification. It states when two non-empty, non-variable
sets are equal by non-deterministically and recursively computing four cases.
These cases implement the \eqref{Ab} and \eqref{Cl} axioms shown in Section
\ref{semantics}. As an example, by applying rule \eqref{eq:ext} to $\{1\} =
\{1,1\}$ we get: ($1 = 1 \land \e = \{1\}) \lor (1 = 1 \land \{1\} = \{1\})
\lor (1 = 1 \land  \e = \{1,1\}) \lor (\e = \{1 \plus \dot N\} \land  \{1 \plus
\dot N\} = \{1\})$, which turns out to be true (due to the second disjunct).

\begin{figure}
\hrule
\begin{flalign}
& \{x  \plus{} A\} = \{y \plus B\} \lfun \label{eq:ext} & \\
& \qquad  (x = y \land A = B)
   \lor (x = y \land \{x \plus A\} = B)
   \lor (x = y \land A = \{y \plus B\})
   \lor (A = \{y \plus \dot N\} \land \{x \plus \dot N\} = B) & \notag \\[2mm]
& \Cup(\{x \plus C\}, A, \dot{B}) \lfun & \label{un:ext1} \\
  & \qquad  \{x \plus C\} = \{x \plus \dot{N_1}\}
      \land x \notin \dot{N_1} \land \dot{B} = \{x \plus \dot N\} & \notag \\
  & \qquad \land (x \notin A \land \Cup(\dot{N_1}, A, \dot N)
  \lor A = \{x \plus \dot{N_2}\}
                 \land x \notin \dot{N_2} \land \Cup(\dot{N_1}, \dot{N_2}, \dot
                 N)) & \notag \\[2mm]
& \Size(\{x \plus A\},m) \lfun
    (x \notin A \land m = 1 + \dot n \land \Size(A,\dot n) \land 0 \leq \dot{n})
   \lor (A = \{x \plus \dot N\} \land x \notin \dot N \land \Size(A,m))
       \label{size:ext}
\end{flalign}
 \hrule
 \caption{\label{f:clpset}Some key rewrite rules inherited from \CLPSET and $\SATCARD$}
\end{figure}

In turn, rule \eqref{un:ext1} is one of the main rules for $\Cup$-constraints.
It deals with $\Cup$-constraints where the first argument is an extensional set
and the last one a variable. Observe that this rule is based on set unification
(i.e., on rule \eqref{eq:ext}). It computes two cases: $x$ does not belong to
$A$, and $x$ belongs to $A$ (in which case $A$ is of the form $\{x \plus
\dot{N_2}\}$ for some set $\dot{N_2}$). In the latter, $x \notin \dot{N_2}$
prevents Algorithm \ref{glob} from generating infinite terms denoting the same
set. The rest of the rewrite rules of \CLPSET can be found in \cite{Dovier00}
and online \cite{calculusBR}.

One of the rewrite rules concerning $\Size$-constraints implemented in
$\SATCARD$ is rule \eqref{size:ext}. It computes the size of any extensional
set by counting the elements that belong to it while taking care of avoiding
duplicates. This means that, for instance, the first non-deterministic choice
for a formula such as $\Size(\{1,2,3,1,4\},m)$ will be:
\[
1 \notin \{2,3,1,4\}
  \land m = 1 + \dot{n} \land \Size(\{2,3,1,4\},\dot{n}) \land 0 \leq \dot{n}
\]
which will eventually lead to a failure due to the presence of  $1 \notin
\{2,3,1,4\}$. This implies that $1$ will be counted
in its second occurrence. Besides, the second choice becomes
$\Size(\{2,3,1,4\},m)$ which is correct given that $\card{\{1,2,3,1,4\}} =
\card{\{2,3,1,4\}}$.

Part of the work of extending $\SATINT$ to integer intervals is to make rewrite
rules such as those shown in Figure \ref{f:clpset} to correctly deal with the
same constraints but when at least one of their arguments is an integer
interval. This is shown in sections \ref{equality}-\ref{sec:size}; and in
Section \ref{discussion} we briefly discuss the new rewrite rules.

\subsection{\label{equality}Rules for $=$-constraints}

The main rewrite rules for $=$-constraints are listed in Figure \ref{f:inteq}.
Rule \eqref{eq:e} is straightforward. Rule \eqref{e:ext} is based on the
identity \eqref{eq:id}. Hence, the rule decides the satisfiability of $[k,m] =
\{y \plus B\}$ by deciding the satisfiability of $\{y \plus B\} \subseteq [k,m]
\land \Size(\{y \plus B\},m-k+1)$. This might seem odd but, due to
\eqref{eq:id}, we know that $\{y \plus B\}$ is actually the interval $[k,m]$.
Observe that \eqref{eq:id} is correctly applied as $k \leq m$ is implicit
because $\{y \plus B\}$ is a non-empty set. In fact, if $m < k$ then $\{y \plus
B\} \subseteq [k,m]$ will fail (see Section \ref{un}). We will discuss the
intuition behind identity \eqref{eq:id} in Section \ref{discussion}.

Rule \eqref{eq:int} takes care of whether the intervals denote the empty set
or not; if not, their corresponding limits must be equal. Rules \eqref{neq:e}
and \eqref{neq:int} are the negations of rules \eqref{eq:e} and \eqref{eq:int},
respectively---rule \eqref{neq:int} includes some Boolean simplification.
Instead of negating rule \eqref{e:ext}, rule \eqref{neq:ext} uses set
extensionality to find out whether or not both sets are different; this is so
for efficiency reasons.

Rules for $=$-constraints where the interval is at the right-hand side have not
been included in the figure since they can be trivially obtained from
those shown in Figure \ref{f:inteq}.

Besides, note that an equality of the form $\dot{A} = [k,m]$ is not rewritten
as a solution for variable $A$ has been found.

\begin{figure}
\hrule
\begin{flalign}
& [k,m] = \e \lfun m < k & \label{eq:e} \\[2mm]
& [k,m] = \{y \plus B\} \lfun \{y \plus B\} \subseteq [k,m] \land
   \Size(\{y \plus B\},m-k+1) & \label{e:ext} \\[2mm]
& [k,m] = [i,j] \lfun
   (k \leq m \land i \leq j \land k = i \land m = j)
   \lor (m < k \land j < i) & \label{eq:int} \\[2mm]
& [k,m] \neq \e \lfun k \leq m & \label{neq:e} \\[2mm]
& [k,m] \neq \{y \plus B\} \lfun
 (\dot{n} \in [k,m] \land \dot{n} \notin \{y \plus B\})
  \lor (\dot{n} \notin [k,m] \land \dot{n} \in \{y \plus B\}) & \label{neq:ext}  \\[2mm]
& [k,m] \neq [i,j] \lfun
 (k \leq m \land (m \neq j \lor j < i \lor k \neq i))
        \lor (i \leq j \land (m \neq j \lor m < k \lor k \neq i)) &
        \label{neq:int}
\end{flalign}
 \hrule
 \caption{\label{f:inteq}Rewrite rules for $=$-constraints involving intervals}
\end{figure}

\subsection{Rules for $\in$-constraints}

The rewrite rules for $\in$-constraints are listed in Figure \ref{f:intin}.
Note that all $\in$-constraints involving integer intervals are rewritten into
integer constraints.

\begin{figure}
\hrule
\begin{flalign}
& x \in [k,m] \lfun k \leq x \leq m & \label{in} \\[2mm]
& x \notin [k,m] \lfun x < k \lor m < x & \label{nin}
\end{flalign}
 \hrule
 \caption{\label{f:intin}Rewrite rules for $\in$-constraints involving intervals}
\end{figure}

\subsection{Rules for $\disj$-constraints}

The main rewrite rules for $\disj$-constraints are listed in Figure
\ref{f:intdisj}. Rule \eqref{disj:ext} uses the identity \eqref{eq:id}; note
that in this case a new variable is introduced. Then, this rule decides the
satisfiability of $\dot{N} \disj A$, instead of $[k,m] \disj A$, because we
know that $\dot{N} = [k,m]$.

Rule \eqref{disj:int} considers all the possible cases when two intervals can
be disjoint. For example, the first case corresponds to the left-hand interval
being the empty set; while the third considers the two cases when both are
non-empty intervals but they are one after the other in the $\num$ line.

Rule \eqref{ndisj:all} uses a general criterion to decide the satisfiability of
$\Ndisj$-constraints by asking for a new variable ($\dot{n}$) to be an
element of both sets. In this way, the $\Ndisj$-constraint is rewritten into
two $\in$-constraints where the first one is dealt with by the
first rule of Figure \ref{f:intin} and the second one by rules introduced
elsewhere \cite{Dovier00}.

Rules for $\disj$-constraints where the interval is at the right-hand side have
not been included in the figure.

\begin{figure}
\hrule
\begin{flalign}
& [k,m] \disj \e \lfun \true & \label{disj:e} \\[2mm]
& \text{If $A$ is a set variable or an extensional set:} & \label{disj:ext} \\
& [k,m] \disj A \lfun
 m < k \lor
    (k \leq m \land \dot{N} \subseteq [k,m] \land \Size(\dot{N},m-k+1) \land
    \dot{N} \disj A) & \notag \\[2mm]
& [k,m] \disj [i,j] \lfun m < k \lor j < i \lor
    (k \leq m \land i \leq j \land (m < i \lor j < k)) & \label{disj:int} \\[2mm]
& [k,m] \Ndisj A \lfun \dot{n} \in [k,m] \land \dot{n} \in A & \label{ndisj:all}
\end{flalign}
 \hrule
 \caption{\label{f:intdisj}Rewrite rules for $\disj$-constraints involving intervals}
\end{figure}

\subsection{\label{un}Rules for $\Cup$-constraints}

The main rewrite rules for $\Cup$-constraints are listed in Figure
\ref{f:intun} and more rules can be found in Appendix \ref{app:rulesun}. In
these figures, $A$ and $B$ represent either variables or extensional set terms;
$C$ represent any set term (including intervals).

Rules \eqref{un:subsete}-\eqref{un:subsetext} are special cases as they,
actually, implement the subset relation (recall that we are using $A \subseteq
B$ just as a synonym of the $\LINT$ constraint $\Cup(A,B,B)$). The first of
these rules is trivial. The second one states that a constraint of the form
$\dot{A} \subseteq [k,m]$ is left unchanged. This is so because these
constraints are always satisfiable by substituting $\dot{A}$ by the empty set.
The importance of this property will be evident in Section \ref{decproc}. The
third rule walks over all the elements of an extensional set until the empty
set or a variable is found. In each step two integer constraints are generated.
Note that when the recursion arrives at the end of the set, rules
\eqref{un:subsete} or \eqref{un:subsetvar} are applied and a fixpoint is
reached.

Rule \eqref{un:3} is based on the identity \eqref{eq:id}. This rule is crucial
as it permits to reconstruct an integer interval from two sets. For instance, this rule covers constraints such as $\Cup(\{x \plus
\dot{A}\},\dot{B},[\dot{k},\dot{m}])$. Note that
the interval disappears from the $\Cup$-constraint. Rule \eqref{un:12} is based
on \eqref{eq:id} as well. In this case we apply the identity twice to transform
the intervals in the $\Cup$-constraint into extensional sets.

Although in rule \eqref{un:123} it would be possible to use \eqref{eq:id}, it
is more efficient to rely on the fact that the union of two integer intervals
is equal to an interval if some linear integer arithmetic conditions hold. As
far as we understand, the same approach cannot be used in the other rules. A
particular case that is true only in $\num$ is considered in the rows labeled
[\textsf{in Z}]. Indeed, for instance, $i \leq m+1$ takes care of a case such
as $\Cup([2,m],[m+1,10],[2,10])$ which does not hold outside of $\num$ as there
are infinite numbers between $m$ and $m+1$.

\begin{figure}
\hrule
\begin{flalign}
& \e \subseteq [k,m] \lfun \true \label{un:subsete} & \\[2mm]
& \dot{A} \subseteq [k,m] \lfun \dot{A} \subseteq [k,m] \text{ (irreducible constraint)} \label{un:subsetvar} & \\[2mm]
& \{y \plus C\} \subseteq [k,m] \lfun k \leq y \leq m \land C \subseteq [k,m]
\label{un:subsetext} & \\[2mm]
& \Cup(A,B,[k,m]) \lfun
 m < k \land A = \e \land B = \e
 \lor (k \leq m \land \dot{N} \subseteq [k,m] \land \Size(\dot{N},m-k+1) \land \Cup(A,B,\dot{N}))  \label{un:3} \\[2mm]
& \Cup([k,m],[i,j],A) \lfun & \label{un:12} \\
& \quad (m < k \land j < i \land A = \e) & \notag \\
& \quad \lor (m < k \land i \leq j \land [i,j] = A) & \notag \\
& \quad \lor (k \leq m \land j < i \land [k,m] = A) & \notag \\
& \quad \lor (k \leq m \land i \leq j
 \land \dot{N}_1 \subseteq [k,m] \land \Size(\dot{N}_1,m-k+1)
 \land \dot{N}_2 \subseteq [i,j] \land \Size(\dot{N}_2,j-i+1)
 \land \Cup(\dot{N}_1,\dot{N}_2,A)) & \notag \\[2mm]
& \Cup([k,m],[i,j],[p,q]) \lfun & \label{un:123} \\
& \quad (m < k \land [i,j] = [p,q]) & \notag \\
& \quad \lor (j < i \land [k,m] = [p,q]) & \notag \\
& \quad \lor (k \leq m \land i \leq j \land k \leq i \land i \leq m+1 \land m \leq j \land p = k \land q = j) & \tag*{[\textsf{in $\num$}]} \\
& \quad \lor (k \leq m \land i \leq j \land k \leq i \land i \leq m+1 \land j < m \land p = k \land q = m) & \notag \\
& \quad \lor (k \leq m \land i \leq j \land i < k \land k \leq j+1 \land m \leq j \land p = i \land q = j) & \tag*{[\textsf{in $\num$}]} \\
& \quad \lor (k \leq m \land i \leq j \land i < k \land k \leq j+1 \land j < m
              \land p = i \land q = m) & \notag
\end{flalign}
 \hrule
 \caption{\label{f:intun}Rewrite rules for $\Cup$-constraints involving intervals}
\end{figure}

Rules for $\Cup$-constraints where $A$ or $B$ are the empty set have not been
included in Figure \ref{f:intun} nor in Appendix \ref{app:rulesun}.

\subsection{\label{sec:size}Rule for $\Size$-constraints}

The rewrite rules for $\Size$-constraints are listed in Figure \ref{f:intsize}.
Note that all $\Size$-constraints involving integer intervals are rewritten
into integer constraints.

\begin{figure}
\hrule
\begin{flalign}
& \Size([k,m],p) \lfun  (m < k \land p = 0) \lor (k \leq m \land p = m-k+1)
& \label{size:size} \\[2mm]
& \Nsize([k,m],p) \lfun  (m < k \land p \neq 0) \lor (k \leq m \land p \neq
m-k+1) & \label{size:nsize}
\end{flalign}
 \hrule
 \caption{\label{f:intsize}Rewrite rules for $\Size$-constraints involving intervals}
\end{figure}

\begin{remark}
The rewrite rules shown in Sections \ref{equality}-\ref{sec:size} can be used
both for variable-intervals and for non-variable-intervals. However, at
implementation level, it is convenient to rewrite some non-variable-intervals
into extensional sets as soon as possible.
\end{remark}

\subsection{\label{discussion}Discussion}

As can be seen by inspecting the rewrite rules, after $\mathsf{step\_loop}$
terminates, integer intervals remain only in constraints of the form $X
\subseteq [p,q]$ where $X$ is a variable and $[p,q]$ is a variable-interval
(i.e., either $p$ or $q$ are variables). Besides, these constraints remain
irreducible (i.e., there are no rewrite rules dealing with them). This makes
the formula returned by the main loop of Algorithm \ref{glob} very similar to
formulas returned at the same point by $\SATCARD$
\cite{DBLP:journals/tplp/CristiaR23}. Precisely, the only difference is the
presence of constraints of the form $X \subseteq [p,q]$.

Note that the new rewrite rules added to $\mathsf{STEP_{S\INT}}$ rewrite set
constraints into integer constraints, whenever possible. We do so because, in
general,  linear integer arithmetic formulas can be solved much more
efficiently than set formulas. When that is impossible we use \eqref{eq:id} to
somewhat trick the solver making it to solve what is, essentially, a $\LCARD$
problem.

Intuitively, \eqref{eq:id} forces an interval to become an extensional set.
Let us see this by applying rule \eqref{un:3} to  an example.

\begin{example}
Consider the following constraint:
\[
\Cup(\{3,x,1\},\{y,5\},[k,m])
\]
which is rewritten by rule \eqref{un:3} into:
\[
N \subseteq [k,m] \land \Size(N,m-k+1) \land \Cup(\{3,x,1\},\{y,5\},N)
\]
where $N$ is intended to be equal to $[k,m]$. At this point rule
\eqref{un:ext1} rewrites the $\Cup$-constraint yielding: $N = \{3,x,1,y,5\}$.
See that $N$ is now an extensional set instead of an interval. Then,
$N$ is substituted by $\{3,x,1,y,5\}$ in the rest of the formula:
\[
\{3,x,1,y,5\} \subseteq [k,m] \land \Size(\{3,x,1,y,5\},m-k+1)
\]
Now, rule \eqref{un:subsetext} is applied several times yielding:
\begin{align*}
& k \leq 3 \leq m
\land k \leq x \leq m
\land k \leq 1 \leq m
\land k \leq y \leq m
\land k \leq 5 \leq m \land \Size(\{3,x,1,y,5\},m-k+1)
\end{align*}
Rule \eqref{size:ext} is applied to the $\Size$-constraint opening several
non-deterministic choices as $\card{\{3,x,1,y,5\}} \in \{3 ,4,5\}$ depending on
the values of $x$ and $y$. In this case, all these choices are encoded as
integer problems.  For instance, when rule \eqref{size:ext} considers the
alternative where $\card{\{3,x,1,y,5\}} = 5$ the formula to solve becomes:
\begin{align*}
& k \leq 3 \leq m
\land k \leq x \leq m
\land k \leq 1 \leq m
\land k \leq y \leq m
\land k \leq 5 \leq m \\
& \land x \neq y \land x \neq 1 \land x \neq 3 \land x \neq 5
  \land y \neq 1 \land y \neq 3 \land y \neq 5 \\
& \land m-k+1 = 5
\end{align*}
Then, $k = 1, m = 5$ and $x$ can be $2$ and $y$ can be $4$ or vice versa.

However, when rule \eqref{size:ext} takes $y = 5 \land x \notin \{3,1,5\}$,
then $\{3,x,1,y,5\}$ becomes $\{3,x,1,5\}$, which cannot be an interval
regardless of the value of $x$ as there are two holes in it (i.e., $2$ and
$4$). In this case $\SATINT$ returns $\false$. \qed
\end{example}

\subsection{Inequality elimination ($\mathsf{remove\_neq}$)}\label{inequalities}


The $\INT$-formula returned by Algorithm \ref{glob} when $\mathsf{STEP_{S\INT}}$
reaches a fixpoint is not necessarily satisfiable.

\begin{example} [Unsatisfiable formula returned by
$\mathsf{STEP_{S\INT}}$]\label{ex:cupcup} The $\INT$-formula:
\begin{equation}\label{eq:ex1}
\Cup(A,B,C) \land \Cup(A,B,D) \land C \neq D
\end{equation}
cannot be further rewritten by any of the rewrite rules of $\mathsf{STEP_{S\INT}}$.
Nevertheless, it is clearly unsatisfiable. \qed
\end{example}

In order to guarantee that $\SATINT$ returns either $\false$ or satisfiable
formulas (see Theorem \ref{satisf}), we still need to remove all inequalities
of the form $\dot{A} \neq t$, where $\dot{A}$ is of sort $\sSet$, occurring as
an argument of $\INT$-constraints based on $\Cup$ or $\Size$. This is performed
(see Algorithm \ref{glob}) by executing the routine \textsf{remove\_neq}, which
applies the rewrite rule described by the generic rule scheme of Figure
\ref{fig:rules_neq_elim}. Basically, this rule exploits set extensionality to
state that non-equal sets can be distinguished by asserting that a fresh
element ($\dot n$) belongs to one but not to the other. Notice that the rule
\eqref{neq_elim:nset} is necessary when $u$ is a non-set term. In this case by
just using rule \eqref{neq_elim:set} we would miss the solution $\dot{A} = \e$.

\begin{figure}
\hrule\vspace{3mm}
 \raggedright
 If $\dot{A}$ is a set variable; $t$ is a set term;
 $u$ is a non-set term; and $\Phi$ is the input formula
 then:
\begin{flalign*}
 \quad\quad & \text{If $\dot{A}$ occurs as an argument of a $\pi$-constraint,
    $\pi \in \{\Cup, \Size\}$, in $\Phi$:} & \\
 & \dot{A} \neq t \lfun (\dot{n} \in \dot{A} \land \dot{n} \notin t) \lor
    (\dot{n} \in t \land \dot{n} \notin \dot{A})
    & \tag{$\neq_\sSet$} \label{neq_elim:set} \\[2mm]
 & \dot{A} \neq u \lfun \true
    & \tag{$\neq_\sUr$} \label{neq_elim:nset}
\end{flalign*}
\hrule
 \caption{Rule scheme for $\neq$-constraint elimination rules}
\label{fig:rules_neq_elim}
\end{figure}

%

\begin{example}[Elimination of $\neq$-constraints]
The $\INT$-formula of Example \ref{ex:cupcup} is rewritten by rule
\eqref{neq_elim:set} to:
\begin{flalign*}
\quad\quad & \Cup(A,B,C) \land \Cup(A,B,D) \land C \neq D \lfun & \\
 & \qquad (\Cup(A,B,C)
  \land \Cup(A,B,D)
  \land \dot{n} \in C \land \dot{n} \notin D) \lor
  (\Cup(A,B,C)
  \land \Cup(A,B,D) \land \dot{n} \notin C \land \dot{n} \in D) &
\end{flalign*}
Then, the $\in$-constraint in the first disjunct is rewritten into a
$=$-constraint (namely, $C = \{\dot{n} \plus \dot{N}\}$), which in turn is
substituted into the first $\Cup$-constraint. This constraint is further
rewritten by rules such as those shown in Figure \ref{f:clpset} and
\cite{Dovier00}, binding either $A$ or $B$ (or both) to a set containing
$\dot{n}$, which in turn forces $D$ to contain $\dot{n}$. This process will
eventually return $\false$, at which point the second disjunct is processed in
a similar way. \qed
\end{example}

\subsection{\label{irreducible}Irreducible constraints}

When no rewrite rule applies to the current $\INT$-formula $\Phi$ and $\Phi$
is not $\false$, the main loop of $\SATINT$ terminates returning $\Phi$ as its
result. This formula can be seen, without loss of generality, as $\Phi_\Set
\land \Phi_\Int$, where $\Phi_\Int$ contains all (and only) integer constraints
and $\Phi_\Set$ contains all other constraints occurring in $\Phi$.

The following definition precisely characterizes the form of atomic constraints
in $\Phi_\Set$.

\begin{definition}[Irreducible formula]\label{def:solved}
Let $\Phi$ be a $\INT$-formula, $A$ and $A_i$ $\INT$-terms of sort $\sSet$, $t$
and $\dot{X}$ $\INT$-terms of sort $\langle \{\sSet,\sUr\} \rangle$, $x$ a
$\INT$-term of any sort, $c$ a variable or a constant integer number, and $k$
and $m$ are terms of sort $\sInt$. A $\INT$-constraint $\phi$ occurring in
$\Phi$ is \emph{irreducible} if it has one of the following forms:
\begin{enumerate}[label=(\roman*), leftmargin=*, widest=viii]
\item \label{i:icfirst} $\dot{X} = t$, and neither $t$ nor $\Phi \setminus \{\phi\}$
contains $\dot{X}$;
\item $\dot{X} \neq t$, and $\dot{X}$ does not occur either in $t$ or
as an argument of any constraint $\pi(\dots)$, $\pi \in \{\Cup,\Size\}$, in
$\Phi$;
\item $x \notin \dot{A}$, and $\dot{A}$ does not occur in $x$;
\item $\Cup(\dot{A}_1,\dot{A}_2,\dot{A}_3)$, where $\dot{A}_1$ and $\dot{A}_2$
are distinct variables;
\item $\dot{A}_1 \disj \dot{A}_2$, where $\dot{A}_1$ and $\dot{A}_2$ are distinct variables;
\item $\Size(\dot{A}, c)$, $c \neq 0$;
\item\label{i:subset} $\dot{A} \subseteq [k,m]$, where $k$ or $m$ are variables.
\end{enumerate}
A $\INT$-formula $\Phi$ is irreducible if it is $\true$ or if all its
$\INT$-constraints are irreducible. \qed
\end{definition}

$\Phi_\Set$, as returned by $\SATINT$ once it finishes its main loop, is an
irreducible formula. This fact can be checked by inspecting the rewrite rules
presented in \cite{Dovier00,DBLP:journals/tplp/CristiaR23} and those given
in this section. This inspection is straightforward as there are no rewrite
rules dealing with irreducible constraints and all non-irreducible form
constraints are dealt with by some rule.

Putting constraints of the form $X \subseteq [k,m]$ aside,  $\Phi_\Set$ is
basically the formula returned by $\SATCARD$. Cristi\'a and Rossi \cite[Theorem
2]{DBLP:journals/tplp/CristiaR23} show that such a formula is always
satisfiable, unless the result is $false$.

It is important to observe that the atomic constraints occurring in $\Phi_\Set$
are indeed quite simple. In particular: \emph{a)} all extensional set terms
occurring in the input formula have been removed, except those occurring as
right-hand sides of $=$ and $\neq$ constraints; and \emph{b)} all integer
interval terms occurring in the input formula have been removed, except those
occurring at the right-hand side of $\subseteq$-constraints. Thus, all
(possibly complex) equalities and inequalities between set terms have been
solved. Furthermore, all arguments of $\Cup$ and $\disj$ constraints are
necessarily simple variables or variable-intervals---and only in constraints of
the form $X \subseteq [k,m]$.

\subsection{\label{after}Checking minimum solutions}
Once the main loop terminates and the returned formula still contains irreducible constraints of form \eqref{i:subset} as given in Definition \ref{def:solved}, Algorithm \ref{glob} calls $\minsol$.
The first step of $\minsol$ is to make the call $\STEPCARD(\Phi_1,Min)$ where $\Phi_1$ is an input and, if the call succeeds, $Min$ represents its output.
As already said, $\STEPCARD$ implements a decision procedure for formulas such as $\Phi_1$, i.e. formulas including LIA, $\Size$, $\Cup$ and $\disj$ constraints \cite{DBLP:journals/tplp/CristiaR23}.
In particular, $\STEPCARD$ uses SWI-Prolog's CLP(Q) library \cite{holzbaur1995ofai} to solve LIA problems by means of the following predicate:
\[
\mathsf{bb\_inf}(\mathit{Vars,Expr,Min,Vert})
\]
which finds a vertex ($\mathit{Vert}$) of the minimum ($\mathit{Min}$) of the expression $\mathit{Expr}$ subjected to the integer constraints present in the constraint store and assuming all the variables in $\mathit{Vars}$ take integer values.
Specifically, $\STEPCARD$ calls $\mathsf{bb\_inf}$ as follows:
\[
\mathsf{bb\_inf}(\mathit{intVars_{\Phi_1},\sum_{i=1}^k m_i,Min,\_})
\]
where $intVars_{\Phi_1}$ are all the integer variables present in $\Phi_1$;
each $m_i$ is the second argument of a $\Size$ constraint present in $\Phi_1$;
and $Min$ is a new variable.
That is, $\mathsf{bb\_inf}$ is called to minimize the sum of all the cardinalities present in $\Phi_1$. 
However, notice that this minimization succeeds only if the integer constraints present in $\Phi_1$ are satisfiable; if not, $\mathsf{bb\_inf}$ simply fails, making $\STEPCARD$ to fail as well.
If this call to $\mathsf{bb\_inf}$ succeeds then $Min$ is bound to an integer number.

If $\STEPCARD(\Phi_1,Min)$ succeeds, $\minsol$ iterates over all the solutions of the integer formula:
\begin{equation}\label{e:clpfd}
0 \leq m_1 \leq Min \land \dots \land 0 \leq m_k \leq Min \land Min = \sum_{i=1}^k m_i
\end{equation}
A solution to the above formula is called \emph{minimum solution}.
This is because the cardinalities present in $\Phi_1$ can only assume values less than or equal to $Min$, which in turn is the minimum value of the sum of all cardinalities.
When \eqref{e:clpfd} is solved $Min$ is an integer number, not a variable.
Furthermore, in the minimum solution $m_1 = c_1, \dots, m_k=c_k$ each $c_i$ is an integer number.

\begin{remark}\label{minsol}
It is easy to see that there is a finite number of minimum solutions. \end{remark}

Therefore, $\minsol$ checks whether or not any minimum solution is a solution of $\Phi_1 \land \Phi_2 \land \Phi_{\subseteq\INT}$ which is the form of the input formula $\Phi$ right after the main loop of Algorithm \ref{glob}.
As soon as a minimum solution is a solution of $\Phi$, $\minsol$ terminates returning $\Phi$.
If no minimum solution is a solution of $\Phi$, $\minsol$ returns $\false$ meaning that $\Phi$ is unsatisfiable.
In the later case we claim that the input formula is unsatisfiable, basically, because any other model would include a minimum solution.
This is proved in Section \ref{decproc}.

Observe that if $m_1 = c_1, \dots, m_k=c_k$ is a minimum solution then all the cardinalities present in $\Phi_1$ are bound to integer numbers.
Then, all $\Size$ constraints in $\Phi_1$ become of the form $\Size(E,c)$ with $E$ a variable and $c$ an integer number.
In this case the following rewrite rule is activated:
\begin{equation}
\text{If $k$ is an integer number: } \qquad
\Size(E,k) \lfun
  E = \{n_1,\dots,n_k\}
  \land ad(n_1,\dots,n_k) \label{size:const3}
\end{equation}
where $n_1,\dots,n_k$ are fresh variables and $ad(n_1,\dots,n_k)$ is a shorthand for $\bigwedge_{i=1}^{k-1} \bigwedge_{j=i+1}^k n_i \neq n_j$ (i.e., all $n_i$ are different from each other).
Hence, $\minsol$ calls $\mathsf{step\_loop}$, instead of $\SATINT$, because there is no need to recursively call $\minsol$ again as there will be no $\Size$ constraints after $\mathsf{step\_loop}$ applies rule \eqref{size:const3}.

\section{\label{decproc}$\SATINT$ is a decision procedure for $\LINT$}

In this section we analyze the soundness, completeness and termination
properties of $\SATINT$. The complete proofs of some theorems can be found in
Appendix \ref{app:proofs}.

The following theorem ensures that each rewriting rule used by $\SATINT$
preserves the set of solutions of the input formula.

\begin{theorem}[Equisatisfiability]\label{equisatisfiable}
Let $\phi$ be a $\INT$-constraint based on symbol $\pi \in \Pi \setminus
\PiInt$, and $\phi \lfun \Phi_1 \lor \dots \lor \Phi_n$ a rewrite rule for
$\pi$-constraints. Then, each solution $\sigma$ of $\phi$\footnote{More
precisely, each solution of $\phi$ expanded to the variables occurring in
$\Phi^i$ but not in $\phi$, so as to account for the possible fresh variables
introduced into $\Phi^i$.} is a solution of $\Phi_1 \lor \dots \lor \Phi_n$,
and vice versa, i.e., $\iS \models \phi[\sigma] \iff \iS \models (\Phi_1 \lor
\dots \lor \Phi_n)[\sigma]$.
\end{theorem}

\begin{proof}
The proof is based on showing that for each rewrite rule the set of solutions
of left and right-hand sides is the same. For those rules dealing with set
terms different from integer intervals the proofs can be found in
\cite{DBLP:journals/jar/CristiaR20} and
\cite{DBLP:journals/tplp/CristiaR23}.

The proof of equisatisfiability for the rules dealing with integer intervals is
as follows. The equisatisfiability of rules \eqref{eq:e},
\eqref{eq:int}-\eqref{disj:e}, \eqref{disj:int}-\eqref{un:subsetext} is trivial
as these rules implement basic results of set theory and integer intervals. The
equisatisfiability of rules \eqref{e:ext}, \eqref{disj:ext}, \eqref{un:3} and
\eqref{un:12} depends on basic facts of set theory and integer intervals (e.g.,
the first branch of rule \eqref{un:3}), and on the application of the identity
\eqref{eq:id}. It is easy to check that \eqref{eq:id} has been consistently
applied on each rule. Note that the same argument can be applied to the rules
included in Appendix \ref{app:rulesun}. The proof of equisatisfiability of rule
\eqref{un:123} can be found in Appendix \ref{app:proofs}.
\end{proof}

The next theorem ensures that, after termination, the whole rewriting
process implemented by $\SATINT$ is correct and complete.

\begin{theorem}[Soundness and completeness]\label{sound&complete}
Let $\Phi$ be a $\INT$-formula and $\Phi^1, \Phi^2,\dots,\Phi^n$ be the
collection of $\INT$-formulas returned by $\SATINT(\Phi)$. Then, every possible
solution of $\Phi$ is a solution of one of the $\Phi^i$ and, vice versa, every
solution of one of these formulas is a solution for $\Phi$.
\end{theorem}

\begin{proof}
According to Definition \ref{irreducible}, each formula $\Phi_i$ returned at
the end of $\SATINT$'s main loop is of the form $\Phi^i_\Set \land
\Phi^i_\Int$, where $\Phi^i_\Set$ is a $\INT$-formula in irreducible form and
$\Phi^i_\Int$ contains all integer constraints encountered during the
processing of the input formula.

As concerns $\Phi^i_\Int$, no rewriting is actually performed on the
constraints occurring in it. Thus the set of solutions is trivially preserved.

Considering also the calls to $\STEPCARD$ and \textsf{gen\_size\_leq}, we
observe that the first is just a check which either returns $\false$ or has no
influence on its input formula, while the second simply adds constraints
entailed by the definition of set cardinality.

As concerns constraints in $\Phi_i^\Set$, we observe that $\SATINT$ is just the
repeated execution of the rewriting rules described in the previous section,
for which we have individually proved equisatisfiability (see Theorem
\ref{equisatisfiable}). No other action of Algorithm \ref{glob} can add or
remove solutions from the input formula.

Thus, the whole $\SATINT$ process preserves the set of solutions of the input
formula.
\end{proof}

\begin{theorem}[Satisfiability of the output formula]\label{satisf}
Any $\INT$-formula different from $\false$ returned by $\SATINT$ is
satisfiable w.r.t. the underlying interpretation structure $\iS$.
\end{theorem}

\begin{proof}

[\textit{sketch}] Given an input formula $\Phi$, containing at least one variable-interval, at the end of the main loop of Algorithm \ref{glob} we have $\Phi \defs \Phi_{\CARD} \land \Phi_{\subseteq[\,]}$, where $\Phi_{\CARD}$ is a $\CARD$-formula and $\Phi_{\subseteq[\,]}$ is a conjunction of constraints of the form $\dot{X} \subseteq [k,m]$ with $k$ or $m$ variables.
As can be seen in Algorithm \ref{glob}, $\STEPCARD$ is called on $\Phi_{\CARD}$ (actually, a sub-formula of it). If $\SATCARD$ finds $\Phi_{\CARD}$ unsatisfiable then $\Phi$ is unsatisfiable.
However, if $\SATCARD$ finds $\Phi_{\CARD}$ satisfiable we still need to check if $\Phi_{\subseteq[\,]}$ does not compromise the satisfiability of $\Phi_{\CARD}$.
To this end, $\SATINT$ iterates over all the minimum solutions of $\Phi_{\CARD}$.
If a minimum solution of $\Phi_{\CARD}$ is a solution of $\Phi$, then $\Phi$ is clearly satisfiable.
Otherwise (i.e., no minimum solution of $\Phi_{\CARD}$ is a solution of $\Phi$), we show that $\Phi$ is unsatisfiable.
That is, if any minimum solution of $\Phi_{\CARD}$ is not a solution of $\Phi$, then any larger solution (w.r.t. the minimum solution) will not be a solution of $\Phi$.
\end{proof}

Now, we can state the termination property for $\SATINT$.

\begin{theorem}[Termination] \label{termination-glob}
The $\SATINT$ procedure can be implemented in such a way that it terminates
for every input $\LINT$ formula.
\end{theorem}

\begin{proof}

[\textit{sketch}] Termination of $\SATINT$ is a consequence of: \emph{a)} termination of
$\SATCARD$ \cite[Theorem 3]{DBLP:journals/tplp/CristiaR23}; \emph{b)} the
individual termination of each new rewrite rule added to
$\mathsf{STEP_{S\INT}}$; and \emph{c)} the collective termination of all the
rewrite rules of $\mathsf{STEP_{S\INT}}$.

Assuming \emph{b)} and \emph{c)}, the same arguments used in \cite[Theorem
3]{DBLP:journals/tplp/CristiaR23} can be applied to Algorithm \ref{glob}.
That is, Algorithm \ref{glob} uses $\mathsf{STEP_{S\INT}}$ instead of the
$\mathsf{STEP_S}$ procedure used by $\SATCARD$ and adds the \textbf{then}
branch after the main loop. $\mathsf{STEP_{S\INT}}$ differs from
$\mathsf{STEP_S}$ in the new rewrite rules introduced in Section \ref{satcard}.
Therefore, it is enough to prove that $\mathsf{STEP_{S\INT}}$ terminates as
$\mathsf{STEP_S}$ does. In turn, this entails to prove \emph{b)} and
\emph{c)}---as done when the termination of \CLPSET and $\SATCARD$ were proved.
\end{proof}

\begin{theorem}[Decidability] \label{decidability}
Given a $\INT$-formula $\Phi$, then $\Phi$ is satisfiable with respect to the
intended interpretation structure $\iS$ if and only if there is a
non-deterministic choice in $\SATINT(\Phi)$ that returns a $\INT$-formula
different from $\false$. Conversely, if all the non-deterministic computations
of $\SATINT(\Phi)$ terminate with $\false$, then $\Phi$ is surely
unsatisfiable. Hence, $\SATINT$ is a decision procedure for $\LINT$.
\end{theorem}

\begin{proof}
Thanks to Theorem \ref{sound&complete} we know that, if $\SATINT$ terminates,
the initial input formula $\Phi$ is equisatisfiable to the disjunction of
formulas $\Phi^1, \Phi^2,\dots,\Phi^n$ non-deterministically returned by
$\SATINT$.
Thanks to Theorem \ref{satisf}, we know that any $\INT$-formula different from
$\false$ returned by $\SATINT$ is surely satisfiable w.r.t. the underlying
interpretation structure $\iS$.
Then, if $\SATINT$ terminates, the initial input formula $\Phi$ is satisfiable
iff the formula $\Phi^1 \lor \dots \lor \Phi^n$ is satisfiable, that is, at
least one of the $\Phi^i$ is different from $\false$.
Thanks to Theorem \ref{termination-glob}, we know that $\SATINT$ terminates for
all admissible $\INT$-formulas. Hence, $\SATINT$ is always able to decide
whether the input formula $\Phi$ is satisfiable or not.
\end{proof}

In Section \ref{expressiveness}, we show several formulas that $\SATINT$ is able
to detect to be unsatisfiable.

Note that many of the rewriting procedures given in the previous section will
stop even when returning relatively complex formulas.

\begin{example}\label{ex:solution}
Assuming all the arguments are variables, $\SATINT$ called on the formula:
\begin{equation*}
\{x \plus A\} = [k,m]
\end{equation*}
will return the following two formulas:
\begin{align*}
& k \leq x \land x \leq m  \land A \subseteq [k,m]  \land x \notin A  \land \Size(A,N_2) \land 1 \leq N_1 \land N_2 = N_1 - 1  \land N_1 = m-k+1 \\[2mm]
& A = \{x \plus N_1\}  \land
k \leq x  \land x \leq m  \land N_1 \subseteq [k,m] \land x \notin N_1 \land
 \Size(N_1,N_3) \land 1 \leq N_2 \land N_3 = N_2 - 1 \land N_2 = m-k+1
\end{align*}
This is so because there is no rewrite rule for
constraints such as $\Size(A,N_2)$ when both arguments are variables. However,
Theorem \ref{satisf} ensures that both formulas are satisfiable. For example,
the first one is satisfiable with $N_1 = 1, N_2 = 0, m = k = x, A = \e$. \qed
\end{example}

\subsection{Complexity of $\SATINT$}

$\SATINT$ strongly relies on set unification. In fact, most rewrite rules
dealing with integer intervals rely on the identity \eqref{eq:id} which,
roughly speaking, forces an interval to become an extensional set and then,
ultimately, to be managed through set unification.

Hence, complexity of our decision procedure strongly depends on complexity of
set unification. As observed in \cite{Dovier2006}, the decision problem for set
unification is NP-complete. A simple proof of the NP-hardness of this problem
has been given in \cite{DBLP:journals/jlp/DovierOPR96}. The proof is based on
representing 3-SAT as a set unification problem; thus, solving the latter in
polynomial time could also be exploited for solving 3-SAT in polynomial time.
Concerning NP-completeness, the algorithm presented here clearly does not
belong to NP since it applies syntactic substitutions. Nevertheless, it would
be possible to encode this algorithm using well-known techniques that avoid
explicit substitutions, maintaining a polynomial time complexity along each
non-deterministic branch of the computation.

Moreover, the implementation of the $\STEPCARD$ procedure requires to
perform, among others, the following actions
\cite{DBLP:journals/tplp/CristiaR23}: compute the set of solutions of
a Boolean formula derived from the irreducible form; and solve an integer
linear programming problem \emph{for each} subset of the Boolean solutions,
which entails to compute the powerset of the Boolean solutions. Both these
problems are inherently exponential in the worse case.

Finally, observe that, $\SATINT$ deals not only with the decision problem for
set unification but also with the associated function problem (i.e., it can
compute solutions for the problem at hand). Since solving the function problem
clearly implies solving the related decision problem, the complexity of
$\SATINT$ can be no better than the complexity of the decision problem for set
unification.

\section{\label{expressiveness}Expressiveness of $\LINT$ , Power of $\SATINT$}

The presence of integer intervals in $\LINT$ is a sensible extension as it can
express many operators and problems that (at least) are hard to express in
$\LCARD$. It is important to observe that all these operators are introduced as
$\INT$-formulas, i.e., as quantifier-free formulas. In this section we explore
the expressiveness of $\LINT$ by means of several examples while we show
examples of what kind of automated reasoning $\SATINT$ is capable of. More
examples can be found in Appendix \ref{app:operators}.

\subsection{Minimum and maximum of a set}
$\LINT$ can express the minimum and maximum of a set as a quantifier-free
formula:
\begin{align}
& \Smin(S,m) \defs m \in S \land S \subseteq [m,\_] \\
& \Smax(S,m) \defs m \in S \land S \subseteq [\_,m]
\end{align}
where ``$\_$'' stands for an anonymous variable as in Prolog. That is, if
$m$ is the minimum of $S$ then $m \in S$ and every other  element in $S$ must
be greater than $m$. This second condition is achieved by stating $S \subseteq
[m,\_]$ because there is no $x \in S$ such that $x < m$ and $x \notin [m,\_]$
given that $m$ is the minimum of $[m,\_]$.

Concerning the automated reasoning that $\SATINT$ can perform, it can easily
prove, for instance, the following propositions by proving that their negations
are unsatisfiable.
\begin{align}
& \Smin(S,m) \implies \forall x \in S: m \leq x \label{pr:min} \\[2mm]
& \Smin(S,m) \land \Smax(S,n) \implies m \leq n
\end{align}
For example, the negation of \eqref{pr:min} is the following $\LINT$ formula:
\begin{equation}
\Smin(S,m) \land x \in S \land x < m
\end{equation}
where $x$ is implicitly existentially quantified.

\subsection{The $i$-th smallest element of set}
The definition of minimum of a set can be generalized to a formula computing
the $i$-th smallest element of a set:
\begin{equation}
\begin{split}
snth(S,i,e) \defs{} & \Cup(Smin,Smax,S) \land Smin \disj Smax \\
& \land
  m \in Smin \land e \in Smin \\ & \land
  \Size(Smin,i) \land
  Smin \subseteq [m,e] \\ & \land
  Smax \subseteq [e + 1,\_]
\end{split}
\end{equation}
The formula partitions $S$ into two disjoint sets $Smin$ and $Smax$.
Intuitively, $Smin$ contains the $i$-th smallest elements of $S$ while $Smax$
contains the rest of $S$. Then $m$ is intended to be the minimum of $S$ which
actually belongs to $Smin$. Then $Smin$ is forced to hold $i$ elements
including $e$ and to be a subset of $[m,e]$. In this way we know that all the
elements of $Smin$ are between $m$ (the minimum of $S$) and $e$ (the $i$-th
smallest element of $S$).  Finally, $Smax$ is forced to be a subset of
$[e+1,\_]$ because otherwise some $x \in S \cap [m,e]$ could be put in $Smax$
and we do not want that. Note, however, that we do not require $e+1 \in Smax$.

\begin{example}
If $\SATINT$ is called as follows it binds $e$ as indicated in each case.
\begin{align*}
& snth(\{7,8,2,14\},1,e) \fun e = 2 \\
& snth(\{7,8,2,14\},2,e) \fun e = 7 \\
& snth(\{7,8,2,14\},3,e) \fun e = 8 \\
& snth(\{7,8,2,14\},4,e) \fun e = 14
\end{align*}
\qed
\end{example}

As the above example shows, $snth$ provides a logic iterator for sets whose
elements belong to a total order. Without $snth$, peeking the ``first'' element
of a set becomes totally non-deterministic as any element of the set can be the
first one. On the contrary, $snth$ provides a deterministic iterator for sets
as the $i$-th smallest element of a set is unique---if its elements belong to a
total order.  Furthermore, as $\SATINT$ is based on constraint programming,
$snth$ allows to compute the index of a given element.

\begin{example}
If $i$ is a variable, then $\SATINT$ will bind $i$ to 3 if the following query
is run: $snth(\{7,8,2,14\},i,8)$. Furthermore, if $x$ is a variable, then
$\SATINT$ will yield conditions for $x$ that make $snth(\{7,8,x,14\},3,8)$
true---specifically, $x < 7$. \qed
\end{example}

Concerning the automated reasoning that $\SATINT$ can perform, it can prove,
for instance, the following propositions by proving that their negations are
unsatisfiable.
\begin{align}
& snth(S,i_1,e_1) \land snth(S,i_2,e_2) \land i_1 < i_2 \implies e_1 < e_2 \\[2mm]
& snth(S,i,e_1) \land snth(S,i+1,e_2) \implies \lnot\exists x \in S: e_1 < x < e_2 \label{pr:snth2}
\end{align}
Observe that \eqref{pr:snth2} is basically a proof of correctness for $snth$.

The $i$-th greatest element of a set can be defined likewise.

\subsection{\label{mnlb_mxub}Partitioning of a set w.r.t. a number}
Consider a set of integer numbers $S$ and any integer $i \notin S$.  The
operator called $mxlb\_mnub$ partitions $S$ into the elements strictly below
$i$ ($L$) and those strictly above $i$ ($U$). Besides, it computes the maximum
of $L$ ($max$) and the minimum of $U$ ($min$), if they exist---either $L$ or
$U$ can be the empty set in some border cases. Hence, $mxlb\_mnub$ computes the
maximum (minimum) of the `lower' (`upper') elements of $S$ w.r.t. $i$.
\begin{equation}
\begin{split}
mxlb\_&mnub(S,i, L,max,U,min) \defs{} \\
  & \Cup(L,U,S) \land L \disj U  \land
  (max < i \land
   smax(L,max)
   \lor
   L = \emptyset
  )  \land
  (i < min \land
   smin(U,min)
   \lor
   U = \emptyset
  )
\end{split}
\end{equation}
It would be possible to remove from the interface of $mxlb\_mnub$ the
arguments $max$ and $min$, and compute them from $L$ and $U$ by calling $smax$
and $smin$. However, since $max$ and $min$ have to be computed inside
$mxlb\_mnub$ to compute $L$ and $U$ it makes sense to include $max$ and $min$
as arguments to avoid a double computation. Besides, note that $mxlb\_mnub$
fails if $i \in S$.

Concerning the automated reasoning that $\SATINT$ can perform, it can prove,
for instance, the following propositions by proving that their negations are
unsatisfiable.
\begin{align}
& mxlb\_mnub(S,i,L,max,U,min) \land smin(S,k) \land i < k \implies L = \emptyset \\[2mm]
& S \subseteq T
\land mxlb\_mnub(S,i,Ls,maxs,Us,mins)
\land mxlb\_mnub(T,i,Lt,maxt,Ut,mint)
  \implies Us \subseteq Ut
\end{align}

$mxlb\_mnub$ is a key operator used in the case study  presented in Section
\ref{casestudy} because it allows to compute the next floor to be served by the
elevator either when moving up or down.

\subsection{\label{max_int}Proper maximal intervals of a set}

Consider a set of integer numbers $S$. It might be useful to find out the
maximal  proper intervals contained in $S$. That is, we look for intervals
$[k,m] \subseteq S$ with $k < m$ such that there is no other interval in $S$
including $[k,m]$. Such intervals may represent, for instance, the longest
continuous paths in a list or graph.
\begin{equation}\label{eq:maxint}
max\_int(S,k,m) \defs
  \Cup([k,m],R,S) \land
  [k,m] \disj R \land
  k < m \land
  k - 1 \notin R \land
  m + 1 \notin R
\end{equation}

\begin{example}
If $\SATINT$ is called on $max\_int(\{5,3,8,2,4,7,1\},k,m)$ it first binds $k$
to 1 and $m$ to 5 and then to 7 and 8. It can also be called on
$max\_int(S,1,5)$ in which case it returns $S = \{1,2,3,4,5 \plus N\}$ plus
constraints forcing $[1,5]$ to be the maximal subinterval in $S$---specifically
$0 \notin N \land 6 \notin N$. \qed
\end{example}

Concerning the automated reasoning that $\SATINT$ can perform, it can prove,for
instance, the following propositions by proving that their negations are
unsatisfiable.
\begin{align}
& max\_int(S,k,m) \land [a,m] \subseteq S \implies k \leq a \\[2mm]
& a < b \land b+2 < c \land \Cup([a,b],[b+2,c],S)
 \land max\_int(S,k,m)
  \implies (k = a \land m = b \lor k = b+2 \land m = c) \label{pr:maxint1}
\end{align}
Note that if in \eqref{pr:maxint1} the left limit of the second interval is
$b+1$ then the maximal interval is $[a,c]$.

If in \eqref{eq:maxint} $k < m$ is removed, then $max\_int$ would return
solutions for the empty interval and for singleton intervals in some cases.
Clearly, $max\_int$ can be generalized to compute only intervals of a minimum
cardinality $c$ by stating $c \leq m - k + 1$ instead of $k < m$.

\section{\label{implementation}\setlog's Implementation of $\SATINT$}

$\LINT$ is implemented by extending the solver provided by the publicly
available tool \setlog (pronounced `setlog') \cite{setlog}. \setlog is a Prolog
program that can be used as a constraint solver, as  a satisfiability solver
and as a constraint logic programming language. It also provides some
programming facilities not described in this paper.

The main syntactic differences between the abstract syntax used in previous
sections and the concrete syntax used in \setlog is made evident by the
following examples.
\begin{example}\label{ex:setlogformulas}
The formula $max\_int$ given in Section \ref{max_int} is written in \setlog as
follows:
\begin{verbatim}
max_int(S,K,M) :-
  un(int(K,M),R,S) &
  disj(int(K,M),R) &
  K < M &
  K1 is K - 1 & K1 nin R &
  M1 is M + 1 & M1 nin R.
\end{verbatim}
where names beginning with a capital letter represent variables, and all others
represent constants and function symbols. As can be seen, \verb+int(K,M)+
corresponds to the integer interval $[K,M]$; \verb+&+ to $\land$;
\verb+disj(int(K,M),R)+ to $[k,m] \disj R$; and \verb+K1 is K - 1 & K1 nin R+
to $k - 1 \notin R$. \qed
\end{example}

In \setlog interval limits and the cardinality of
$\Size$-constraints can only be variables or constants. Besides, the
extensional set constructor $\w$ is encoded as \verb+{_/_}+. All
this is shown in the following example.

\begin{example}
A formula such as:
\[
\Cup(\{x \plus A\},B,[k+1,m]) \land \Size(A,p) \land \Size(B,p-3)
\]
is encoded in \setlog as follows:
\begin{verbatim}
un({X/A},B,int(K1,M)) & K1 is K + 1 & size(A,P) & size(B,P3) & P3 is P - 3.
\end{verbatim}
In other words, constraints such as \verb.un(A,B,int(K + 1,M)). or
\verb+size(A,P - 3)+ make \setlog to output an error message. \qed
\end{example}

More examples on how to use \setlog are given in Section \ref{casestudy}.

\subsection{Rewrite rules for subset, intersection and difference}

As we have said in Section \ref{derived}, subset, intersection ($\Cap$) and
difference ($diff$) are definable in terms of union and disjoint. This means
that when a formula including subset, intersection or difference is processed
it is first transformed into a $\LINT$ formula by substituting these operators
by union and disjoint. This works well from the theoretical perspective but in
practice it leads to performance penalties.

Therefore, we extend the implementation of $\SATINT$ in \setlog by including
rewrite rules for subset, intersection and difference---this follows the
implementation of \CLPSET and $\SATCARD$. As with the primitive constraints,
the rewrite rules for subset, intersection and difference are based either on
simple mathematical results (e.g., $\Cap([k,m],\e,A) \lfun A = \e$); on the application of the identity \eqref{eq:id}; or on integer
arithmetic constraints---such as rule \eqref{un:123} for $\Cup$-constraints. As
an example, Figure \ref{f:intdiff} shows rule \eqref{diff:12} for
$diff$-constraints where we rely on integer arithmetic constraints as much as
possible until the last case where rule \eqref{un:12} is called---which in turn
is based on \eqref{eq:id}. This last case can be graphically represented over
the $\num$ line as follows:
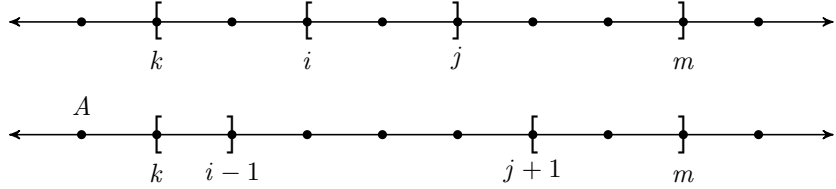
\begin{center}
\begin{tikzpicture} [shorten <=1pt,>=stealth',semithick]
\draw[<->] (-1,0) -- (10,0);
\foreach \x in {0,1,2,3,4,5,6,7,8,9} {
  \node[draw,circle,fill=black,inner sep=1pt]
       (A\x) at (\x cm,0) {};
  }
\node[label={{\LARGE\bf [}}] (k1) at (29pt,-12pt) {};
\node[label={$k$}] (k) at (1,-25pt) {};
\node[label={{\LARGE\bf ]}}] (m1) at (227pt,-12pt) {};
\node[label={$m$}] (m) at (8,-25pt) {};
\node[label={{\LARGE\bf [}}] (i1) at (86pt,-12pt) {};
\node[label={$i$}] (i) at (3,-25pt) {};
\node[label={{\LARGE\bf ]}}] (j1) at (142pt,-12pt) {};
\node[label={$j$}] (j) at (5,-25pt) {};
\draw[<->] (-1,-1.5) -- (10,-1.5);
\node[label=$A$] (a) at (0,-1.5) {};
\foreach \x in {0,1,2,3,4,5,6,7,8,9} {
  \node[draw,circle,fill=black,inner sep=1pt]
       (A\x) at (\x cm,-1.5) {};
  }
\node[label={{\LARGE\bf [}}] (k1) at (1.02,-55pt) {};
\node[label={$k$}] (k) at (1,-68pt) {};
\node[label={{\LARGE\bf ]}}] (m1) at (7.98,-55pt) {};
\node[label={$m$}] (m) at (8,-68pt) {};
\node[label={{\LARGE\bf ]}}] (i2) at (1.97,-55pt) {};
\node[label={$i-1$}] (i3) at (2,-68pt) {};
\node[label={{\LARGE\bf [}}] (j2) at (6.02,-55pt) {};
\node[label={$j+1$}] (j3) at (6.1,-68pt) {};
\end{tikzpicture}
\end{center}
Observe that, in spite that rule \eqref{diff:12} calls rule \eqref{un:12}, it
does not cause termination problems as the rules for $\Cup$-constraints do not
call rules for $diff$-constraints.

\begin{figure}
\hrule
\begin{flalign}
& diff([k,m],[i,j],A) \lfun & \label{diff:12} \\
& \quad m < k \land A = \e & \notag \\
& \quad\lor k \leq m \land j < i \land A = [k,m] & \notag \\
& \quad\lor k \leq m \land i \leq j \land m < i \land A = [k,m] & \notag \\
& \quad\lor k \leq m \land i \leq j \land j < k \land A = [k,m] & \notag \\
& \quad\lor i \leq k \leq m \leq j \land A = \e & \notag \\
& \quad\lor k \leq i \leq m \leq j \land A = [k,i-1] & \notag \\
& \quad\lor i \leq k \leq j \leq m \land A = [j+1,m] & \notag \\
& \quad\lor k \leq i \leq j \leq m \land \Cup([k,i-1],[j+1,m],A) & \notag
\end{flalign}
 \hrule
 \caption{\label{f:intdiff}A rewrite rule for $diff$-constraints involving intervals}
\end{figure}

\subsection{A memoizing schema}

\setlog processes formulas by rewriting one constraint at a time.
As we have seen, some rewrite rules apply the identity \eqref{eq:id} to
substitute an integer interval by a new variable plus some constraints. In this
way, if a given integer interval appears in two or more constraints which are
rewritten by rules that apply \eqref{eq:id}, that interval will be substituted
by different variables. The following example illustrates this.

\begin{example}\label{ex:memo}
When the last alternative of rule \eqref{un:12} is applied to the following
formula:
\begin{verbatim}
un(int(K,M),int(I,J),A) & un(int(I,J),int(K,M),B) & A neq B
\end{verbatim}
the result is a formula such as:
\begin{verbatim}
un(W,X,A) & un(Y,Z,B) & A neq B &
K =< M & I =< J &
subset(W,int(K,M)) & size(W,P1) & P1 is M - K + 1 &
subset(X,int(I,J)) & size(X,P2) & P2 is J - I + 1 &
subset(Y,int(I,J)) & size(Y,P3) & P3 is J - I + 1 &
subset(Z,int(K,M)) & size(Z,P4) & P4 is M - K + 1
\end{verbatim}
Note how, for instance, \verb+int(K,M)+ has been substituted by \verb+W+ and
\verb+Z+.  Clearly, this formula implies \verb+W = Z+ but \setlog will deduce
this after many rewriting steps. \qed
\end{example}

This rewriting schema makes some formulas unnecessarily complex and, in
general, degrades \setlog efficiency when dealing with integer intervals. In
order to avoid this problem we have implemented a memoizing schema that keeps
track of what variable has been used to substitute a given integer interval.
Then, when an integer interval is about to be substituted, \setlog looks up if
it has already been substituted and in that case it reuses the variable used in
the first substitution. In this way, any given integer interval is always
substituted by the same variable.

\begin{example}
With the memoizing schema, the formula of Example \ref{ex:memo} is
rewritten as follows:
\begin{verbatim}
un(W,X,A) & un(X,W,B) & A neq B &
subset(W,int(K,M)) & size(W,P1) & P1 is M - K + 1 &
subset(X,int(I,J)) & size(X,P2) & P2 is J - I + 1
\end{verbatim}
where  it is evident that the
two $\Cup$-constraints share the same variables \verb+W+ and \verb+X+.
\qed
\end{example}

This memoizing schema makes a linear search over a list every time
\eqref{eq:id} is applied. Since set solving (especially cardinality solving) can
be exponential in time, the memoizing schema produces a sensible gain in
efficiency. It may degrade the efficiency only in very specific cases which are
nonetheless solved quickly. \setlog with the memoizing schema solves the
formula of Example \ref{ex:memo} ten times faster than without it.

\subsection{\label{empirical}An initial empirical evaluation}

Several in-depth empirical evaluations provide evidence that \setlog is able to
solve non-trivial problems
\cite{DBLP:journals/jar/CristiaR20,DBLP:conf/RelMiCS/CristiaR18,DBLP:journals/jar/CristiaR21a,CristiaRossiSEFM13};
in particular as an automated verifier of security properties
\cite{DBLP:journals/jar/CristiaR21,DBLP:journals/jar/CristiaR21b}.

As far as we know there are no benchmarks for a language like $\LINT$. There
are a couple of benchmarks for languages performing \emph{only} interval
reasoning, i.e., intervals cannot be mixed with sets and not all set operators
are supported. These languages are meant to solve specific verification
problems---for instance, model-checking of interval temporal logic
\cite{DBLP:journals/cacm/Allen83}.

Then, besides the case study  presented in Section \ref{casestudy}, we have
gathered 60 $\LINT$ formulas stating properties of the operators defined in
Section \ref{expressiveness}---including all the properties stated in that
section. \setlog solves\footnote{These problems and the case study of Section
\ref{casestudy} were solved on a HP Elitebook with 11\textsuperscript{th} Gen Intel(R) Core(TM) i7-1165G7 at 2.80GHz with 32 Gb of main memory, running
Linux Ubuntu 24.04.4 LTS, SWI-Prolog 9.3.18-2-g304db6161 and \setlog 4.9.9-2d.} all the problems
in 30.6 seconds thus averaging 0.51 seconds per problem. Only 9 problems take
more than 1 second of which only formula \eqref{pr:snth2} takes more than 5 seconds---it takes 11 seconds.

The benchmark can be found here \url{https://www.clpset.unipr.it/SETLOG/setlog-intervals.zip}, along with instructions
to reproduce our results.

\section{\label{casestudy}Case Study}

In this section we present a case study using  the implementation of $\LINT$
and $\SATINT$ in \setlog. The intention of the case study is to show that
\setlog is useful in practice when it comes to solve
problems involving integer intervals, especially concerning automatically
discharging proof obligations. Here we present a simplified version to make the
presentation more amenable. The \setlog program can be found
in the file \texttt{lift.pl} located in the same URL indicated above.

The case study is based on the elevator problem. That is, there is an elevator
receiving service requests from the floors and from inside it. The control
software should move the elevator up and down according to the requests it must
serve. The key requirement is that the elevator shall move in one direction as
long as there are requests that can be served in that direction. In particular,
the elevator shall serve first the nearest request in the direction of
movement. We have included  in the case study requirements about stopping the
elevator, opening and closing the door, etc.

\subsection{A \setlog program}

The resulting program consists of a 180 LOC \setlog program implementing seven
operations of the elevator control software (add a request, serve next request,
close the door, start the elevator, pass by a floor, stop the elevator and open
the door). The number of LOC might look too small but this is due, in part, to
the fact that many complex operations can be written very compactly by using
set theory. The \setlog code shown below corresponds to one of the main
operations of the program, namely \verb+nextRequestUp+, which computes the next
request to be served when the elevator is moving up.
\begin{verbatim}
nextRequestUp(Lift,Lift_) :-
  Lift = [F,Nf,D,C,M,R] &
  M = up &
  diff(R,{F},R1) &
  mxlb_mnub(R1,F,_,_,Ub,Nf_) &
  (Ub neq {} & M_ = M
   or
   Ub = {} & Nf_ = Nf & M_ = none
  ) &
  Lift_ = [F,Nf_,D,C,M_,R].
\end{verbatim}

\verb+Lift+ and \verb+Lift_+ represent the before and after states of the
elevator, respectively. That is, \verb+Lift_+ plays the same role as $Lift'$ in
notations such as B and Z. As can be seen, \verb+Lift+ is a 6-tuple where each
variable holds part of the state of the system: \verb+F+ represents the floor
which the elevator is currently passing by; \verb+Nf+ is the next floor to be
served; \verb+D+ represents the elevator's door (\verb+open+ or \verb+closed+);
\verb+C+ specifies whether the elevator is \verb+moving+ or \verb+halted+;
\verb+M+ is the direction of movement (\verb+up+, \verb+down+ or \verb+none+);
and \verb+R+ is the set of requests to be served. The next state is updated in
the last line by unifying \verb+Lift_+ with \verb+[F,Nf_,D,M_,R]+ where some of
the variables are different from those used in the initial tuple. The same
naming convention is used: \verb+M+ (\verb+M_+) is the current (next) direction
of movement.

\verb+nextRequestUp+ computes the next floor to be served (\verb+Nf_+) by
calling \verb+mxlb_mnub+ (see Section \ref{mnlb_mxub}) but only paying
attention to the upper bounds (\verb+Ub+) of the requests to be served w.r.t.
the current floor. Clearly, if the elevator is moving up then it should keep
that direction unless there are no more requests in that direction. Then,
\verb+nextRequestUp+ distinguishes two cases: \verb+Ub+ is not empty and so it
takes the minimum of \verb+Ub+ as the next floor to be served by putting
\verb+Nf_+ as the second component of \verb+Lift_+; or \verb+Ub+ is empty and
so the direction of movement is changed to \verb+none+. In this last case, the
software can change the direction to \verb+down+ if \verb+R+ is not empty (this
is done by an operation called \verb+nextRequestNone+ which becomes enabled
when \verb+M = none+). The three alternatives to compute the next floor are
assembled in one operation:
\begin{verbatim}
nextRequest(Lift,Lift_) :-
  nextRequestNone(Lift,Lift_) or nextRequestUp(Lift,Lift_) or  nextRequestDown(Lift,Lift_).
\end{verbatim}

\subsection{Simulations}
With this code we can run simulations to evaluate how the system works by
setting the current state and calling some operations. Simulations are encoded
as \setlog formulas. For example:
\begin{verbatim}
Lift = [3,3,closed,halted,up,{2,5,8,1,0}] & nextRequest(Lift,Lift_).
\end{verbatim}
returns the next state:
\begin{verbatim}
Lift_ = [3,5,closed,halted,up,{2,5,8,1,0}]
\end{verbatim}
We can see that the next floor to be served is 5 because it is the nearest
requested floor going up. Instead, if the elevator is moving down we get 2 as
the next floor to be served:
\begin{verbatim}
Lift = [3,3,closed,halted,down,{2,5,8,1,0}] & nextRequest(Lift,Lift_).

Lift_ = [3,2,closed,halted,down,{2,5,8,1,0}]
\end{verbatim}
It is also possible to call more than one operation\footnote{\texttt{addRequest(0,20,Lift1,4,Lift2)} states that the elevator runs in a building with floors numbered from 0 to 20 and a request to the 4th floor is added.}:
\begin{verbatim}
Lift1 = [3,3,closed,halted,up,{2,5,8,1,0}] &
addRequest(0,20,Lift1,4,Lift2) & nextRequest(Lift2,Lift3).

Lift2 = [3,3,closed,halted,up,{2,5,8,1,0,4}], Lift3 = [3,4,closed,halted,up],{2,5,8,1,0,4}]
\end{verbatim}
Note that \setlog produces the state trace. Since a request to the 4th floor
has been added, the next floor to be served is the 4th. If \verb+addRequest+
adds the 12th floor, the next floor to be served would be the 5th.

\subsection{Automated proofs}
Being able to run simulations on
the code is good to have a first idea on how the system works, but this cannot
guarantee the program is correct. If we need stronger evidence on the
correctness of the program we should try to prove some properties true of it.
In this section we show that we can use the \emph{same representation} of the
control software and the \emph{same tool} that we have used to run simulations
(i.e., \setlog), also to automatically prove properties of it.

In this context, one of the canonical class of properties to be proved are
state invariants. Therefore, we have stated 7 state invariants and we have used
\setlog to prove that all the state operations preserve all of them. This
amounts to automatically discharge 49 invariance lemmas---plus 7 proving that
the initial state satisfies all the state invariants. \setlog discharges all
these proof obligations in about half of a second---on a standard laptop
computer. The state invariants and the lemmas are encoded with 900 LOC of
\setlog code.

As with simulations, state invariants and invariance lemmas are encoded as
\setlog formulas. Just to give an idea of what is this all about, in Figure
\ref{f:inv} we reproduce one state invariant and in Figure \ref{f:lemma} one
invariance lemma. Both figures present a mathematical encoding and the
corresponding \setlog encoding. As we have explained, logical implication has
to be encoded as disjunction. In Figure \ref{f:lemma}, the \setlog encoding
corresponds to the negation of the mathematical encoding, given that \setlog
proves a lemma by proving that its negation is unsatisfiable. Then,
\verb+n_liftInv3+ is the Boolean negation of \verb+liftInv3+.

\begin{figure}
 \hrule
\[
liftInv3(Lift) \defs
Lift.Di = up \land Lift.C = moving \implies Lift.F \leq Lift.Nf
\]
\begin{verbatim}
liftInv3(Lift) :- Lift = [F,Nf,D,C,Di,R] & (Di neq up or C neq moving or F =< Nf).
\end{verbatim}
 \hrule
\caption{\label{f:inv}A typical state invariant}
\end{figure}

\begin{figure}
 \hrule
\[
nextRequest\_PI\_inv3 \defs
liftInv3(Lift) \land nextRequest(Lift,Lift') \implies liftInv3(Lift')
\]
\begin{verbatim}
nextRequest_PI_inv3 :- liftInv3(Lift) & nextRequest(Lift,Lift_) & n_liftInv3(Lift_).
\end{verbatim}
 \hrule
\caption{\label{f:lemma}A typical invariance lemma}
\end{figure}

It is important to observe that if a formula that is supposed to be a lemma is
run on \setlog but it happens not to be valid, then \setlog will return a
counterexample.

\begin{example}
 If in \verb+liftInv3+ the inequality \verb+F =< Nf+ is changed to \verb+F < Nf+,
then operation \verb+passFloor+ will not preserve that invariant because this
operation can increment \verb+F+ in one. Then, the formula to be run
is:
\begin{verbatim}
Lift = [F,Nf,D,C,Di,R] &
(Di neq up or C neq moving or F < Nf) &
passFloor(MinF,MaxF,Lift,Lift_) &
Lift_ = [F_,Nf_,D_,C_,Di_,R_] &
Di_ = up & C_ = moving & F_ >= Nf_.
\end{verbatim}
In this case \setlog returns a counterexample stating that:
\begin{verbatim}
F_ = Nf, Nf_ = Nf, Nf is F + 1
\end{verbatim}
That is, initially, \verb+F < Nf+ but \verb+passFloor+ increments \verb+F+ in
one and `assigns' this value to \verb+F_+ while leaving \verb+Nf+ unchanged.
Then, after \verb+passFloor+ has executed the current floor can be equal to the
next floor to be served.

More concrete counterexamples can be obtained by changing the default integer
solver to CLP(FD) (\setlog command \verb+int_solver(clpfd)+), although
in that case the user has to give values for \verb+MinF+ and \verb+MaxF+ and
has to state that \verb+F+ ranges between those limits, i.e.,
\verb+F in int(MinF,MaxF)+. \qed
\end{example}

The case study provides evidence that \setlog can deal with verification
problems involving integer intervals by providing
simulation and proving capabilities over the same representation of the system.

\section{\label{relwork}Related Work}

Tools such as Atelier B \cite{Mentre00} and ProB \cite{Leuschel00} are very
good in performing automated reasoning and a variety of analysis over B
specifications. B specifications are based on a set theory including $\LINT$.
We are not aware of these tools implementing a decision procedure for that
fragment of set theory. Integrating \setlog into these tools would constitute a
promising line of work.

There are a number of works dealing with constraints admitting integer
intervals but where their limits are constants (e.g.,
\cite{DBLP:journals/constraints/HarveyS03,DBLP:journals/constraints/AptZ07}).
For this reason, in these approaches full automated reasoning is not possible.
Some of these approaches accept non-linear integer constraints. In general,
they aim at a different class of problems, notably constraint programming to
solve hard combinatorial problems.

Allen \cite{DBLP:journals/cacm/Allen83} defines an interval-based temporal
logic which later on has been widely studied (e.g.,
\cite{DBLP:journals/iandc/BozzelliMMP20,DBLP:journals/jacm/KrokhinJJ03}). In
Allen's logic, intervals are defined over the real line. This logic can be
expressed as a relation algebra \cite{DBLP:conf/ijcai/Ladkin87}. The algebra
has been proved to be decidable. Answer Set Programming (ASP) is closely
related to CLP. Janhunen and Sioutis \cite{DBLP:conf/inap/JanhunenS19} use ASP
to solve problems expressed in Allen's interval algebra over the rational line.

A possible abstract representation of an array is as a function over the
integer interval $[1,n]$ where $n$ is the array length\footnote{The interval
can also be $[0,n-1]$ depending on the convention used for array indexes.};
then, $[1,n]$ becomes the array domain. The point here is that research on the
theory of arrays sometimes needs to solve problems about integer intervals. For
instance, Bradley et al. \cite{DBLP:conf/vmcai/BradleyMS06} define a decision
procedure for a fragment of the theory of arrays that allows to reason about
properties holding for the array components with indexes in $[k,m]$. \setlog
might help in that context. For example, if we have proved that the components
of array $A$ with indexes in two sets, $I$ and $J$, verify some property we may
want to prove that the property holds for the whole array by proving that
$\Cup(I,J,[1,n])$. It will be interesting to investigate whether or not the
solving capabilities of \setlog concerning partial functions
\cite{DBLP:journals/jar/CristiaR20} combined with the results of this paper
could deal with the decidable fragment found by Bradley et al.

Also motivated by research on formal verification of programs with arrays,
Eriksson and Parsa \cite{DBLP:conf/padl/ErikssonP20} define a domain specific
language for integer interval reasoning. They pay particular attention to
partition diagrams that divide an array domain into several (disjoint) integer
intervals. Different properties of an array are true of each interval in the
diagram. The DSL has been prototyped in the Why3 platform.

Integer interval reasoning is also used in static program analysis---sometimes
in connection with arrays. Su and Wagner \cite{DBLP:journals/tcs/SuW05} present
a polynomial time algorithm for a general class of integer interval
constraints. In this work, the authors define a lattice of intervals and
interval constraints are defined over the partial order of the lattice. The
language allows for infinite intervals where limits can be $\pm\infty$.

After reviewing these works we can draw two conclusions. First, some have
addressed the problem of reasoning about languages where the only available
sets are intervals
\cite{DBLP:journals/cacm/Allen83,DBLP:conf/padl/ErikssonP20,DBLP:journals/tcs/SuW05}.
Without other kinds of sets it is not possible to define and reason about the
operators discussed in Section \ref{expressiveness}. Instead, $\LINT$ can, for
instance, express all the interval relations defined by Allen, over the integer
line. The second conclusion is that no approach to interval reasoning seems to
be rooted in set theory. On the contrary, the extension of \setlog to integer
intervals seamlessly integrate interval reasoning with set reasoning allowing
to freely combine sets with intervals. This makes our approach more general and
coherent, perhaps paying the price of a reduced efficiency when it comes to
specific problems.


\section{\label{conclusions}Conclusions and Future Work}

We have presented a language and a decision procedure for the algebra of finite
sets extended with cardinality constraints and finite integer intervals. As far
as we know, this is the first time that a language with such features is proved
to be decidable. The implementation of the decision procedure as part of the
\setlog tool has also been presented. Initial empirical evidence showing that
\setlog could be used in practice is available.

\setlog supports complex relational constraints such as composition, domain and
domain restriction \cite{DBLP:journals/jar/CristiaR20}. When integer intervals
are combined with relational and cardinality constraints it is possible to
model (finite) arrays:
\[
array(A,n) \defs 0 < n \land \Size(A,n) \land dom(A,D) \land D \subseteq [1,n] \land \Size(D,n)
\]
where $n$ is the length of array $A$. That is, $A$ is a function (set of
ordered pairs) whose domain is the integer interval $[1,n]$---note that the
cardinality of $A$ and its domain $D$ is the same and $D = [1,n]$ due to
\eqref{eq:id}. Therefore, our next step is to investigate what are the
decidable fragments concerning arrays. This would yield a powerful tool to work
on the automated verification of programs with arrays.

\bibliographystyle{elsarticle-num}
\bibliography{/home/mcristia/escritos/biblio}

\appendix

\section{\label{app:rulesun}More Rewrite Rules for $\Cup$-constraints}

Below some more rewrite rules included in $\mathsf{STEP_{S\INT}}$ for
$\Cup$-constraints are listed. $\mathsf{STEP_{S\INT}}$ also includes rules
symmetric to rules \eqref{un:1} and \eqref{un:13} when the interval $[k,m]$ is
the second argument and not the first.

\vspace{3mm}

\noindent
\parbox{\textwidth}{%
\hrule
\begin{flalign}
& \Cup([k,m],A,B) \lfun
 (m < k \land A = B)
 \lor (k \leq m \land \dot{N} \subseteq [k,m] \land \Size(\dot{N},m-k+1) \land \Cup(\dot{N},A,B)) \label{un:1} \\[2mm]
& \Cup([k,m],A,[i,j]) \lfun & \label{un:13} \\
& \quad j < i \land [k,m] = A = \e & \notag \\
& \quad \lor i \leq j \land m < k \land A = [i,j] & \notag \\
& \quad \lor k \leq m \land i \leq j
 \land \dot{N}_1 \subseteq [k,m] \land \Size(\dot{N}_1,m-k+1)
 \land \dot{N}_2 \subseteq [i,j] \land \Size(\dot{N}_2,j-i+1)
 \land \Cup(\dot{N}_1,A,\dot{N}_2) & \notag
\end{flalign}
 \hrule
}

\section{\label{app:proofs}Proofs}
This section contains the proofs of some results referred in the main document.

The proof of the fundamental identity \eqref{eq:id}.
\begin{lemma}
If $A$ is any finite set, then:
\[
\forall k,m \in \num:
k \leq m \implies (A = [k,m] \iff A \subseteq [k,m] \land \card{A} = m - k + 1)
\]
\end{lemma}

\begin{proof}
Assuming $k \leq m$ the proof of:
\[
A = [k,m] \implies A \subseteq [k,m] \land \card{A} = m - k + 1
\]
is trivial. Note that without assuming $k \leq m$ the result is not true because if $m < k$ then $\card{[k,m]} = 0$ while $m - k + 1$ is not necessarily $0$.

Now, assuming $k \leq m$ the proof of:
\[
A \subseteq [k,m] \land \card{A} = m - k + 1 \implies A = [k,m]
\]
is also simple because knowing that $\card{[k,m]} = m - k + 1$ then $A$ and $[k,m]$ are two sets of the same cardinality with one of them being a subset of the other. This implies the two sets are indeed the same. Then, $A = [k,m]$.
\end{proof}

Proof of equisatisfiability of rule \eqref{un:123}.

\begin{lemma}
\begin{flalign}
& \Cup([k,m],[i,j],[p,q]) \iff & \notag \\
& \quad m < k \land [i,j] = [p,q] & \tag*{{\footnotesize[1st]}} \\
& \quad \lor j < i \land [k,m] = [p,q] & \tag*{{\footnotesize[2nd]}} \\
& \quad \lor k \leq m \land i \leq j \land k \leq i \land i \leq m+1 \land m \leq j \land p = k \land q = j & \tag*{{\footnotesize[3th]}} \\
& \quad \lor k \leq m \land i \leq j \land k \leq i \land i \leq m+1 \land j < m \land p = k \land q = m & \tag*{{\footnotesize[4th]}} \\
& \quad \lor k \leq m \land i \leq j \land i < k \land k \leq j+1 \land m \leq j \land p = i \land q = j & \tag*{{\footnotesize[5th]}} \\
& \quad \lor k \leq m \land i \leq j \land i < k \land k \leq j+1 \land j < m \land p = i \land q = m & \tag*{{\footnotesize[6th]}}
\end{flalign}
\end{lemma}

\begin{proof}
The proof of the first two branches is trivial.

The proof of the 3rd and 4th branches is symmetric to the proof of the 5th and 6th branches. This symmetry comes from considering whether $[k,m]$ is at the left of $[i,j]$ or vice versa. In turn, this is expressed by stating $k \leq i$ in the 3rd and 4th branches and $i < k$ in the 5th and 6th.

Then, we will only prove the 3rd and 4th branches with the help of a geometric argument over the $\num$ line. Indeed if $[k,m]$ is at the left of $[i,j]$ we have the following cases.

The first case is depicted as follows.
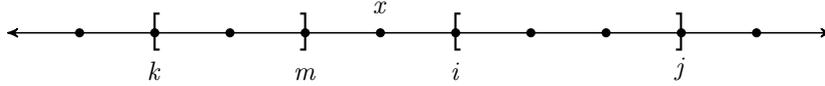
\begin{center}
\begin{tikzpicture} [shorten <=1pt,>=stealth',semithick]
\draw[<->] (-1,0) -- (10,0);
\foreach \x in {0,1,2,3,4,5,6,7,8,9} {
  \node[draw,circle,fill=black,inner sep=1pt]
       (A\x) at (\x cm,0) {};
  }
\node[label={{\LARGE\bf [}}] (k1) at (1,-12pt) {};
\node[label={$k$}] (k) at (1,-25pt) {};
\node[label={{\LARGE\bf ]}}] (m1) at (3,-12pt) {};
\node[label={$m$}] (m) at (3,-25pt) {};
\node[label={$x$}] (m) at (4,0) {};
\node[label={{\LARGE\bf [}}] (i1) at (5,-12pt) {};
\node[label={$i$}] (i) at (5,-25pt) {};
\node[label={{\LARGE\bf ]}}] (j1) at (8,-12pt) {};
\node[label={$j$}] (j) at (8,-25pt) {};
\end{tikzpicture}
\end{center}
In this case the union of $[k,m]$ and $[i,j]$ cannot yield an integer interval because there is a hole ($x$) in between them. This case is avoided by stating $i \leq m+1$ in both branches.

The second case is depicted as follows.
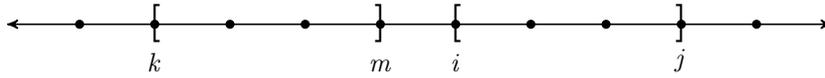
\begin{center}
\begin{tikzpicture} [shorten <=1pt,>=stealth',semithick]
\draw[<->] (-1,0) -- (10,0);
\foreach \x in {0,1,2,3,4,5,6,7,8,9} {
  \node[draw,circle,fill=black,inner sep=1pt]
       (A\x) at (\x cm,0) {};
  }
\node[label={{\LARGE\bf [}}] (k1) at (1,-12pt) {};
\node[label={$k$}] (k) at (1,-25pt) {};
\node[label={{\LARGE\bf ]}}] (m1) at (4,-12pt) {};
\node[label={$m$}] (m) at (4,-25pt) {};
\node[label={{\LARGE\bf [}}] (i1) at (5,-12pt) {};
\node[label={$i$}] (i) at (5,-25pt) {};
\node[label={{\LARGE\bf ]}}] (j1) at (8,-12pt) {};
\node[label={$j$}] (j) at (8,-25pt) {};
\end{tikzpicture}
\end{center}
In this case the union of $[k,m]$ and $[i,j]$ is equal to $[k,j]$; that is, $p
= k$ and $q = j$. This case is covered in the 3rd branch when $i = m+1$. Note
that this case is only valid on the $\num$ line because there are no
\emph{integer} numbers between $m$ and $m+1$ (i.e., $i$).

The third case is depicted as follows.
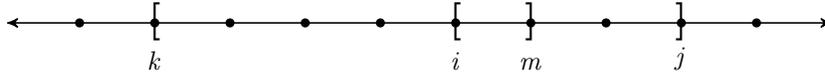
\begin{center}
\begin{tikzpicture} [shorten <=1pt,>=stealth',semithick]
\draw[<->] (-1,0) -- (10,0);
\foreach \x in {0,1,2,3,4,5,6,7,8,9} {
  \node[draw,circle,fill=black,inner sep=1pt]
       (A\x) at (\x cm,0) {};
  }
\node[label={{\LARGE\bf [}}] (k1) at (1,-12pt) {};
\node[label={$k$}] (k) at (1,-25pt) {};
\node[label={{\LARGE\bf ]}}] (m1) at (6,-12pt) {};
\node[label={$m$}] (m) at (6,-25pt) {};
\node[label={{\LARGE\bf [}}] (i1) at (5,-12pt) {};
\node[label={$i$}] (i) at (5,-25pt) {};
\node[label={{\LARGE\bf ]}}] (j1) at (8,-12pt) {};
\node[label={$j$}] (j) at (8,-25pt) {};
\end{tikzpicture}
\end{center}
In this case the union of $[k,m]$ and $[i,j]$ is again equal to $[k,j]$; that
is, $p = k$ and $q = j$. This case is covered also in the 3rd branch when $i <
m$.

The fourth and last case is depicted as follows.
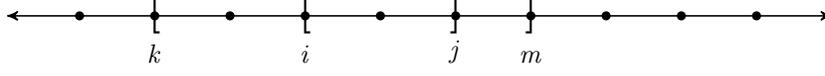
\begin{center}
\begin{tikzpicture} [shorten <=1pt,>=stealth',semithick]
\draw[<->] (-1,0) -- (10,0);
\foreach \x in {0,1,2,3,4,5,6,7,8,9} {
  \node[draw,circle,fill=black,inner sep=1pt]
       (A\x) at (\x cm,0) {};
  }
\node[label={{\Large\bf [}}] (k1) at (1,-12pt) {};
\node[label={$k$}] (k) at (1,-25pt) {};
\node[label={{\Large\bf ]}}] (m1) at (6,-12pt) {};
\node[label={$m$}] (m) at (6,-25pt) {};
\node[label={{\Large\bf [}}] (i1) at (3,-12pt) {};
\node[label={$i$}] (i) at (3,-25pt) {};
\node[label={{\Large\bf ]}}] (j1) at (5,-12pt) {};
\node[label={$j$}] (j) at (5,-25pt) {};
\end{tikzpicture}
\end{center}
In this case the union of $[k,m]$ and $[i,j]$ is equal to $[k,m]$; that is, $p
= k$ and $q = m$. This case is covered in the 4th branch by stating $j < m$.
\end{proof}


The next is the proof of Theorem \ref{satisf}.

\begin{proof}
As we have analyzed at the end of Section \ref{discussion}, when there are no integer intervals in the input formula $\Phi$, $\SATINT$ behaves exactly as $\SATCARD$.
Hence, for such formulas the theorem is proved elsewhere \cite{DBLP:journals/tplp/CristiaR23}.

Now, consider an input formula $\Phi$ containing at least one variable-interval.
Hence, at the end of the main loop of Algorithm \ref{glob} we have $\Phi \defs \Phi_{\CARD} \land \Phi_{\subseteq[\,]}$, where $\Phi_{\CARD}$ is a $\CARD$-formula and $\Phi_{\subseteq[\,]}$ is a conjunction of constraints of the form $\dot{X} \subseteq [k,m]$ with $k$ or $m$ variables.
As can be seen in Algorithm \ref{glob}, $\Phi_{\CARD}$ is divided into $\Phi_1$ and $\Phi_2$.

If $\STEPCARD(\Phi_1,Min)$ fails it is because $\Phi_1$ is unsatisfiable.
If $\Phi_1$ is unsatisfiable, $\Phi$ is unsatisfiable just because $\Phi \defs \Phi_1 \land \Phi_2 \land \Phi_{\subseteq[\,]}$.

On the other hand, if $\STEPCARD(\Phi_1,Min)$ succeeds Algorithm \ref{glob} iterates over all the minimum solutions making the following call in each iteration:
\begin{equation}
\mathsf{step\_loop}(\Phi_1 \land \Phi_2 \land \Phi_{\subseteq\INT} \land m_1 = c_1 \land \dots \land m_k=c_k) \label{eq:call}
\end{equation}
where $m_i$ are the second arguments of the $\Size$-constraints in $\Phi_1$ and $c_i$ the corresponding minimum values as given by the current minimum solution.
We can think that initially $\mathsf{STEP_{S\INT}}$ substitutes
$m_i$ by $c_i$ in the rest of the formula:
\begin{equation*}
(\Phi_1 \land \Phi_2 \land \Phi_{\subseteq[\,]})[\forall i \in 1 \upto n: m_i \gets c_i]
\end{equation*}
This means that all $\Size$-constraints become $\Size(\dot{A}_j,c_j)$,
with $c_j$ constant. In a second iteration $\mathsf{STEP_{S\INT}}$ applies rule \eqref{size:const3}:
\begin{equation}
\Size(\dot{A}_i,c_i) \lfun \dot{A}_i = \{n_1^i,\dots,n_{c_i}^i\} \land
\bigwedge_{p=1}^{v_i} \bigwedge_{q=1}^{v_i} n_p^i \neq n_q^i
\tag*{\footnotesize[if $p \neq q$]}
\end{equation}
where $n_1^i,\dots,n_{c_i}^i$ are all new variables.
Hence, all $\Size$-constraints are eliminated from the formula. 
Furthermore, if any such $\dot{A}_i$ is at the left-hand side of a constraint in $\Phi_{\subseteq[\,]}$, $\mathsf{STEP_{S\INT}}$ will perform the following substitution:
\begin{equation*}
\Phi_{\subseteq[\,]}[\forall j \in 1 \upto k: \dot{A}_j \gets \{n_1^j,\dots,n_{c_j}^j\}]
\end{equation*}

Now, $\mathsf{STEP_{S\INT}}$ applies rule \eqref{un:subsetext} to each
constraint in $\Phi_{\subseteq[\,]}$ where the substitution has been performed.
Then, all those constraints where \eqref{un:subsetext} is applied will be
rewritten into a conjunction of integer constraints:
\begin{equation}\label{eq:subsetint}
\{n_1^i,\dots,n_{v_i}^i\} \subseteq [k^i,m^i] \lfun \bigwedge_{p=1}^{v_i} k^i
\leq n_p^i \leq m^i
\end{equation}
This implies that all the remaining constraints in
$\Phi_{\subseteq[\,]}$ are of the form $\dot{A} \subseteq [k,m]$ and
there is no $\Size$-constraint with $\dot{A}$ as first argument.

At this point, the resulting formula is either $\false$ or a conjunction of:
\emph{i)} a \CLPSET formula in solved form; \emph{ii)} a linear integer formula
and; \emph{iii)} a conjunction of constraints of the form $\dot{A} \subseteq
[k,m]$, with $k$ or $m$ variables, and where there is no $\Size$-constraint
with $\dot{A}$ as first argument. The case when the resulting formula is
$\false$ will be analyzed below. Then, let us consider the case when the
resulting formula is a conjunction of \emph{i}-\emph{iii}. The linear integer
formula is satisfiable because $\mathsf{STEP_{S\INT}}$ calls CLP(Q) at each
iteration. The \CLPSET formula in solved form is satisfiable due to results
presented elsewhere \cite{Dovier00}. Furthermore, a solution can be obtained by
substituting all set variables by the empty set. Finally, the third conjunct is
also satisfiable by substituting all set variables by the empty set. That is, a
constraint of the form $\dot{A} \subseteq [k,m]$ has as a solution $\dot{A} =
\e$. This solution does not conflict with other occurrences of $\dot{A}$ in the
\CLPSET formula because, there, it is substituted by the empty set, too.
Therefore, we can conclude that the input formula, $\Phi$, is satisfiable.
Basically, we have proved that the minimum solution of $\Phi_1$ is a solution
of $\Phi_1 \land \Phi_2 \land \Phi_{\subseteq[\,]}$.

Now, we will analyze the case when \eqref{eq:call} returns $\false$ for some minimum solution.
Clearly, if this call answers $\false$ it might be due to the minimum solution
of $\Phi_1$ not being a solution for $\Phi_1 \land \Phi_2 \land
\Phi_{\subseteq[\,]}$.
If $\Phi_1 \land \Phi_2 \land \Phi_{\subseteq[\,]}$ is
satisfiable there should be other solutions.
We will prove that if a minimum solution of $\Phi_1$ is not a solution of $\Phi_1 \land \Phi_2 \land \Phi_{\subseteq[\,]}$, then no \emph{larger solution} is a solution of the formula being analyzed.
If $\bigwedge_{i=1}^n m_i = c_i$ is a minimum solution, then a \emph{larger solution} assigns $c_i$ to $m_i$ for all $i \in [1,n]$ except for at least one $m_{i_0}$ ($i_0 \in [1,n]$) that is bound to $c_{i_0} + 1$---a larger solution can assign any $c'_{i_0}$ to $m_{i_0}$ with $c'_{i_0} > c_{i_0}$ but it suffices to consider $c_{i_0} + 1$.

Then, we assume that the following is $\false$:
\begin{equation}\label{eq:notminsol}
\Phi_1 \land (\bigwedge_{i=1}^n m_i = c_i) \land \Phi_2 \land \Phi_{\subseteq[\,]}
\end{equation}
On the other hand:
\begin{equation}\label{eq:notminsol2}
\Phi_1 \land (\bigwedge_{i=1}^n m_i = c_i) \land \Phi_2
\end{equation}
is satisfiable due to the properties of $\SATCARD$---that is,  if
$\bigwedge_{i=1}^n m_i = c_i$ is a solution of $\Phi_1$ is also a solution of
$\Phi_1 \land \Phi_2$. Then, \eqref{eq:notminsol} becomes unsatisfiable due to
interactions with $\Phi_{\subseteq[\,]}$.

As we have seen in the first part of the proof, when a minimum solution is
propagated in $\Phi_{\subseteq[\,]}$ rule \eqref{un:subsetext} is applied, at
least, to some of its $\subseteq$-constraints. Hence, a $\subseteq$-constraint
in $\Phi_{\subseteq[\,]}$ either remains unchanged or is rewritten as in
\eqref{eq:subsetint}. If the $\subseteq$-constraint remains unchanged, then a
larger solution than the minimum one will not make this $\subseteq$-constraint
to be rewritten, and so a larger solution will not change the satisfiability of
\eqref{eq:notminsol}. Now, if the $\subseteq$-constraint is rewritten as in
\eqref{eq:subsetint} we know that:
\begin{enumerate}
  \item All the $n_p^i$ are of sort $\sInt$.
  \item All the $n_p^i$ are such that $k^i \leq n_p^i \leq m^i$.
  \item There are at least $c_i$ integer numbers in $[k^i,m^i]$---because
  all the $n_p^i$ are different from each other.
\end{enumerate}
Therefore, if \eqref{eq:notminsol} is unsatisfiable when \eqref{eq:notminsol2}
is satisfiable it must be because $\Phi_1 \land \Phi_2$ implies at least one of
the following (for some $x$, $p$ and $i$ such that $A_i \subseteq [k^i,m^i]$):
\begin{enumerate}
  %
  \item $x \in \dot{A}_i \land (x < k^i \lor m^i < x)$, when $\dot{A}_i \subseteq [k^i, m^i]$ making the whole formula trivially unsatisfiable.

  \item $\Size(\dot{A}_i,q_i) \land m^i-k^i+1 < q_i$ 

  $m^i-k^i+1 < q_i$, $\Size(\dot{A}_i,q_i) \land \dot{A}_i \subseteq [k^i,m^i]$ is trivially unsatisfiable. 
  %

  %
  \item $0 \leq p \land (\bigwedge_{j=1}^{p+1} x_j \notin \dot{A}_i \land
  k^i \leq x_j \leq m^i) \land \Size(\dot{A}_i,m^i-k^i+1 - p)$
  \begin{enumerate}
    \item\label{i:p+1} $\bigwedge_{j=1}^{p+1} k^i \leq x_j \leq m^i$ implies
    that there are $p+1$ elements in $[k^i,m^i]$.
    \item\label{i:m-k+1-p} $\Size(\dot{A}_i,m^i-k^i+1-p) \land \dot{A}_i \subseteq [k^i,m^i]$
    implies that there are $m^i-k^i+1-p$ elements in $[k^i,m^i]$.
    \item $\bigwedge_{j=1}^{p+1} x_j \notin \dot{A}_i$ implies that the $p+1$
    elements of \ref{i:p+1} are disjoint from the $m^i-k^i+1-p$ elements of \ref{i:m-k+1-p}.
    \item Then, there are $p+1+m^i-k^i+1-p = m^i-k^i+2$ elements in $[k^i,m^i]$
    when its cardinality is $m^i-k^i+1$.
  \end{enumerate}
  \end{enumerate}

If any of the above makes \eqref{eq:notminsol} unsatisfiable, any solution
larger than the considered minimum solution will not change that because (1)-(3) are implied by $\Phi_1
\land \Phi_2$ and so any of its solutions will imply the same.
\end{proof}


The next is the proof of Theorem \ref{termination-glob}.

\begin{proof}
Termination of $\SATINT$ is a consequence of: \emph{a)} termination of
$\SATCARD$ \cite[Theorem 3]{DBLP:journals/tplp/CristiaR23}; \emph{b)} the fact that there's a finite number of minimum solutions; \emph{c)} the
individual termination of each new rewrite rule added to
$\mathsf{STEP_{S\INT}}$; and \emph{d)} the collective termination of all the
rewrite rules of $\mathsf{STEP_{S\INT}}$.

Assuming \emph{c)} and \emph{d)}, the same arguments used in \cite[Theorem
3]{DBLP:journals/tplp/CristiaR23} can be applied to Algorithm \ref{glob}.
That is, Algorithm \ref{glob} uses $\mathsf{STEP_{S\INT}}$ instead of the
$\mathsf{STEP_S}$ procedure used by $\SATCARD$ and adds the \textbf{then}
branch after the main loop. $\mathsf{STEP_{S\INT}}$ differs from
$\mathsf{STEP_S}$ in the new rewrite rules introduced in Section \ref{satcard}.
Therefore, it is enough to prove that $\mathsf{STEP_{S\INT}}$ terminates as
$\mathsf{STEP_S}$ does. In turn, this entails to prove \emph{c)} and
\emph{d)}---as done when the termination of \CLPSET and $\SATCARD$ were proved.

Concerning \emph{c)} (i.e., individual termination), observe that in Figures
\ref{f:inteq}-\ref{f:intun} and Appendix \ref{app:rulesun} the only self
recursive rule is \eqref{un:subsetext}. However, termination of
\eqref{un:subsetext} is guaranteed as, at some point, it arrives at the set
part of the extensional set. At this point \eqref{un:subsetext} becomes
disabled and \eqref{un:subsete} or \eqref{un:subsetvar} are executed. Note that
\eqref{un:subsetext} is not mutually recursive.

Concerning \emph{d)} (i.e., collective termination), we can observe the
following:
\begin{enumerate}
  \item\label{i:satcard}
Rewrite rules of $\SATCARD$ do not call the new rewrite rules as $\LCARD$ does
not admit integer intervals. Then, if one of the new rewrite rules generates a
constraint dealt with by the old rewrite rules, the latter will not call the former.
Then, termination is proved in what concerns the interaction between new and
old rewrite rules.
  \item
Rules \eqref{eq:e}, \eqref{eq:int}, \eqref{neq:e}, \eqref{neq:int}, \eqref{in},
\eqref{nin}, \eqref{disj:e}, \eqref{disj:int}, \eqref{un:subsete},
\eqref{un:123}, \eqref{size:size} and \eqref{size:nsize}, produce only integer
constraints or $\true$. All these constraints are processed by CLP(Q) which
only generates integer constraints. Then, termination is proved.
  \item\label{i:innin}
Rules \eqref{in} and \eqref{nin} trivially terminate. This is important because
other rules in Figures \ref{f:inteq}-\ref{f:intun} produce constraints of the
form $\cdot \in \i$ or $\cdot \notin \i$. So when that happens, termination is
proved.
  \item\label{i:subset}
Rules \eqref{un:subsete}-\eqref{un:subsetext} are either non-recursive or self
recursive, but they never make a recursion on another rule. This is important
because these are the rules called when the identity \eqref{eq:id} is applied.
Then, when \eqref{un:subsete}-\eqref{un:subsetext} are called, termination is
proved.
  \item\label{i:ext}
Rule \eqref{e:ext} generates a $\Size$-constraint in which case
\ref{i:satcard} applies; and a constraint of the form $\cdot
\subseteq \i$ in which case \ref{i:subset} applies.
Then, termination is proved.
  \item
Rule \eqref{neq:ext} calls rules \eqref{in} and \eqref{nin}, in which case
\ref{i:innin} applies. Then, termination is proved.
  \item
Rule \eqref{disj:ext} generates a $\Size$ and a $\disj$ constraints in which
case \ref{i:satcard} applies; and a constraint of the form $\cdot
\subseteq \i$ in which case \ref{i:subset} applies. Then,
termination is proved.
  \item
Rule \eqref{ndisj:all} generates a constraint of the form $\cdot \in \i$
in which case \ref{i:innin} applies. Then, termination is
proved.
  \item
Rules \eqref{un:3} and \eqref{un:12} generate equality, $\Size$ and $\Cup$
constraints in which case \ref{i:satcard} applies. They also generate
constraints of the form $\i = \cdot$ in which case
\ref{i:ext} applies; and constraints of the form $\cdot \subseteq \i$
in which case \ref{i:subset} applies. Then, termination is
proved.
\end{enumerate}
\end{proof}

\section{\label{app:operators}Two Operators Definable as $\LINT$ Formulas}
In this section we show two more operators definable as $\LINT$ formulas.

$lb\_ub(S,Lb,Ub)$ holds if the set $S$ is
partitioned into two equal halves one containing the lower elements ($Lb$) and
the other the higher elements ($Ub$). In particular, if the cardinality
of $S$ is odd, the predicate does not hold.
\begin{equation}
lb\_ub(S,Lb,Ub) \defs
       \Cup(Lb,Ub,S) \land Lb \disj Ub
       \land Lb \subseteq [\_,m] \land \Size(Lb,k)
       \land m < n
       \land Ub \subseteq [n,\_] \land \Size(Ub,k)
\end{equation}

$ssucc(A,x,y)$ holds if $y \in A$ is the successor of
$x \in A$.
\begin{equation}
\begin{split}
ssu&cc(A,x,y) \defs \\
    & x < y
     \land A = \{x,y \plus B\} \land x \notin B \land y \notin B
     \land \Cup(Inf,Sup,B) \land Inf \disj Sup
     \land Inf \subseteq [\_,x - 1]
     \land Sup \subseteq [y + 1,\_]
\end{split}
\end{equation}

\end{document}